\documentclass[journal]{IEEEtran}
\usepackage{amsmath,amsfonts,amssymb}
\usepackage[dvips]{graphicx}
\usepackage{verbatim}
\usepackage{slashbox}
\usepackage{xcolor}
\usepackage{subfigure}
\usepackage{bm}
\usepackage{algorithm} 
\usepackage{algorithmic} 
\usepackage{mathptmx}
\usepackage{balance}
\usepackage[T1]{fontenc}
\usepackage{cite}
\IEEEoverridecommandlockouts

\newtheorem{theorem}{Theorem}

\newcommand{\figwidth}{9.2}

\begin{document}

\title{Iterative Eigenvalue Decomposition and Multipath-Grouping Tx/Rx Joint Beamformings for Millimeter-Wave Communications}

\author{Zhenyu Xiao,~\IEEEmembership{Member,~IEEE,}
        Xiang-Gen Xia,~\IEEEmembership{Fellow,~IEEE,}
        Depeng Jin,~\IEEEmembership{Member,~IEEE,}
        and Ning Ge,~\IEEEmembership{Member,~IEEE}

\thanks{This work was partially supported by the National Natural Science Foundation of China (NSFC) under grant nos. 61201189, 91338106, and 61231013, the Fundamental Research Funds for the Central Universities under grant nos. YWF-14-DZXY-007, WF-14-DZXY-020 and YMF-14-DZXY-027, National Basic Research Program of China under grant no. 2011CB707000, Foundation for Innovative Research Groups of the National Natural Science Foundation of China under grant no. 61221061, and the National Science Foundation (NSF) of USA under Grant CCF-0964500.}
\thanks{Dr. Zhenyu Xiao is with the School of
Electronic and Information Engineering, Beihang University, Beijing 100191, P.R. China.}
\thanks{Prof. Xiang-Gen Xia is with the Department of Electrical and Computer Engineering, University of Delaware, Newark, DE 19716, USA.}
\thanks{Profs. Depeng Jin and Ning Ge are Department of
Electronic Engineering, Tsinghua University, Beijing 100084, P.R.
China.}

\thanks{Corresponding Author: Zhenyu Xiao (xiaozy06@gmail.com).}
}
\markboth{IEEE TRANSACTIONS ON WIRELESS COMMUNICATIONS, VOL. X, NO.
x, xxx 20xx}{Shell \MakeLowercase{\textit{et al.}}}

\maketitle
\begin{abstract}
We investigate Tx/Rx joint beamforming in millimeter-wave communications (MMWC). As the multipath components (MPCs) have different steering angles and independent fadings, beamforming aims at achieving array gain as well as diversity gain in this scenario. A sub-optimal beamforming scheme is proposed to find the antenna weight vectors (AWVs) at Tx/Rx via iterative eigenvalue decomposition (EVD), provided that full channel state information (CSI) is available at both the transmitter and receiver. To make this scheme practically feasible in MMWC, a corresponding training approach is suggested to avoid the channel estimation and iterative EVD computation. As in fast fading scenario the training approach may be time-consuming due to frequent training, another beamforming scheme, which exploits the quasi-static steering angles in MMWC, is proposed to reduce the overhead and increase the system reliability by multipath grouping (MPG). The scheme first groups the MPCs and then concurrently beamforms towards multiple steering angles of the grouped MPCs, so that both array gain and diversity gain are achieved. Performance comparisons show that, compared with the corresponding state-of-the-art schemes, the iterative EVD scheme with the training approach achieves the same performance with a reduced overhead and complexity, while the MPG scheme achieves better performance with an approximately equivalent complexity.
\end{abstract}


\section{Introduction}
Millimeter-wave communications (MMWC) refer to the communications with a carrier frequency in the millimeter-wave (MMW) band, typically tens or even hundreds of GHz. MMWC has a great commercial potential, and has attracted growing attentions owing to its abundant frequency spectrum resource, which enables a much higher capacity than the existing communications below 6 GHz, e.g., the wireless local area network (WLAN) and the cellular mobile communications \cite{wangjunbo2009ant,wangjb2013dyn,wangg2014energy,wangg2014accurate}. This trend is demonstrated by the recent progresses in research and development of 60 GHz WLAN \cite{Rapp_2010_60GHz_general,wang_2011_MMWCS,Park_2008_mmw_wlan_cha&feas,Park_2010_11ad,Xia_2011_60GHz_Tech,xiaozhenyu2013div}, as well as MMW mobile broadband communications \cite{khan_2011,pi_2011_MMW_intro}.

However, a significant challenge that may affect the potential promising prospect of MMWC is high propagation attenuation resulting from the high carrier frequency. To remedy this, antenna arrays can be adopted at both the source and destination devices, where a single radio-frequency (RF) chain (or a single data stream) is tied to the antenna array, to exploit array gains via appropriate single-layer beamforming \cite{Rapp_2010_60GHz_general,wang_2011_MMWCS,Park_2008_mmw_wlan_cha&feas,Park_2010_11ad,Xia_2011_60GHz_Tech}. In general, a single beam is shaped at both the transmitter and receiver, steering towards each other or a specific reflector to achieve array gain \cite{xia_2008_prac_ante_traning,xia_2008_prac_SDMA}. However, such single-direction beamforming may be not robust due to fading or blocking \cite{Park_2012_beam_diversity,xiaozhenyu2013div}. In fact, the multipath components (MPCs) in MMWC have different steering angles and independent fadings \cite{maltsev_2010,geng_2009,Rappa_2002_60GHz_indoor_ch,jacob_2011,rapp_2011_MMW,Rapp_2012_cellular_MMW}. Thus, it becomes necessary to achieve not only array gain, but also diversity gain to increase system reliability, via beamforming.

For single-layer beamforming, it is known that, provided the channel state information (CSI) at both ends, the optimal antenna weight vectors (AWVs) at Tx/Rx can be found under well-known performance criteria, e.g., maximizing receive signal-to-noise ratio (SNR) \cite{tang_2005,wang_2009_beam_codebook,libin2013,nsenga_2009}. In narrow band systems, it is well known that the optimal receive and transmit AWVs are the left and right principal singular vectors of the channel matrix, respectively \cite{tang_2005,TseFundaWC}. However, in wideband MMWC which experiences frequency-selective channels, it is difficult to find a solution of the optimal transmit and receive AWVs. In \cite{nsenga_2009}, Nsenga, et al. have proposed a sub-optimal scheme exploiting eigenvalue decomposition (EVD) and Schmidt decomposition within a high-dimensional space tensed by the transmit and receive AWVs. Although less CSI is required than full CSI, the channel estimation is still time-consuming, and the computations of EVD and Schmidt decomposition in the tensor space are also complicated due to the high dimension, which may limit the practical application of this scheme. On the other hand, in fast fading scenario \footnote{It is noted that a fast/slow fading channel means that the estimation of multi-antenna channel is required frequently/non-frequently in this paper. For instance, in fast fading channel, the estimation may be required each several packets, but in slow fading channel, the estimation may be required each tens of packets.}, these sub-optimal schemes that require CSI at both ends become infeasible due to frequent channel estimation or training, which are time-consuming. In order to reduce overhead and meanwhile increase system reliability, the scheme proposed by Park and Pan in \cite{Park_2012_beam_diversity} can be adopted, which utilizes the quasi-static steering angles of the MPCs in MMWC, and concurrently beamforms along multiple steering angles at both ends to achieve diversity gain in addition to array gain. This scheme is simple to implement and achieves full diversity, but it may be not efficient enough in array gain, and may be infeasible when the number of MPCs is larger than that of the antennas in either ends, because in such a case a solution of AWV may not exist.

The contributions of this paper are twofold. First, a new sub-optimal beamforming scheme is proposed, which finds the AWVs via iterative EVD (IEVD), provided that full CSI is available at both the transmitter and receiver. To make this sub-optimal scheme practically feasible in MMWC, a corresponding training approach is suggested to avoid the channel estimation and iterative EVD computation. The convergence analysis is also provided. Furthermore, in fast fading scenario, a multipath-grouping (MPG) based beamforming scheme is proposed to reduce overhead and increase system reliability. The scheme first groups the MPCs and then concurrently beamforms towards multiple steering angles of the grouped MPCs, so that both array gain and diversity gain are achieved. Owing to the MPG operation, the scheme guarantees a solution of AWV even when the number of MPCs is greater than that of the antennas at both ends. Pairwise-error probability (PEP) and diversity analyses are given. Performance comparisons show that the proposed IEVD scheme achieves the same performance as Nsenga's scheme, but has a lower overhead and complexity by exploiting the training approach. In addition, the MPG scheme outperforms Park and Pan's scheme in array gain with an approximately equivalent complexity.

The rest of this paper is organized as follows. In Section II we introduce the system and the channel models. In Section III we present the IEVD scheme and its training approach, and conduct the convergence analysis. In Section IV, we first give a brief description of the scheme proposed by Park and Pan, and then introduce the MPG scheme. Afterwards, we analyze the PEP and diversity performance. In Section V, we present performance comparison. The conclusions are drawn lastly in Section VI.

\section{System and Channel Models}
Without loss of generality, we consider an MMWC system with
half-wave spaced uniform linear arrays (ULAs) of $n_{\rm{t}}$ and $n_{\rm{r}}$ elements
at the transmitter and receiver, respectively \cite{wang_2009_beam_codebook}, as shown in Fig. \ref{fig:system}. A single RF chain is tied to the ULA at the transmitter and receiver. At the transmitter, a single data stream is transmitted from multiple weighted antenna elements, and at the receiver, signals from multiple antenna elements are weighted and combined to shape a single signal stream. It is noted that the system is half duplex or time-division duplex, i.e., a data stream can also be transmitted from the receiver to the transmitter in the same frequency band but at different time. According to the reported results of channel measurement for MMWC \cite{Xia_2011_60GHz_Tech,maltsev_2010}, only reflection contributes to generating MPCs besides the line-of-sight (LOS) component; scattering and diffraction effects are little due to the extremely small wave length of MMWC. Thus, the MPCs in MMWC have a directional feature, i.e., different MPCs have different physical transmit steering angles $\phi_{{\rm{t}}\ell}$ and receive steering angles $\phi_{{\rm{{\rm{r}}}}\ell}$, as shown in Fig. \ref{fig:system}. In fact,  \cite{moraitis_2007,Park_2012_beam_diversity,sayeed_2011,sayeed_2007} have reported such channel models. It is noted that although only reflection generates significant MPCs, the number of MPCs may be not always small, because there may be many good reflectors in MMWC. For instance, the walls, floor, ceiling, and metal objects are good reflectors for the indoor MMWC \cite{maltsev_2010,geng_2009,Rappa_2002_60GHz_indoor_ch,jacob_2011}, and the buildings are good reflectors for the outdoor MMWC \cite{rapp_2011_MMW,Rapp_2012_cellular_MMW}. Moreover, the second-order reflection components may also have a significant strength \cite{maltsev_2010,geng_2009,Rappa_2002_60GHz_indoor_ch,jacob_2011}.

\begin{figure}[t]
\begin{center}
  \includegraphics[width=7.2 cm]{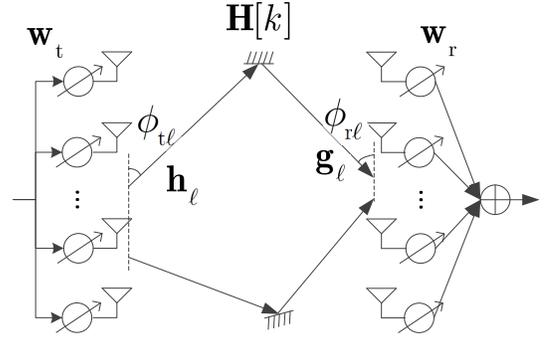}
  \caption{Illustration of the system.}
  \label{fig:system}
\end{center}
\end{figure}

In order to achieve a high transmission speed, the bandwidth of MMWC is generally large. For instance, in both 60 GHz WLAN \cite{Rapp_2010_60GHz_general,wang_2011_MMWCS,Park_2008_mmw_wlan_cha&feas,Park_2010_11ad,Xia_2011_60GHz_Tech,xiaozhenyu2013div} and MMW mobile broadband communications \cite{khan_2011,pi_2011_MMW_intro}, the signal bandwidth is basically greater than 1 GHz, which means the equivalent symbol duration is less than 1 ns. That is to say, a difference of only 3 m in path distance would result in a path delay difference of about $3/(3\times 10^8\times 10^{-9})=10$ symbol durations. As in MMWC the propagation distance may be tens of meters for WLAN or hundreds of meters for mobile communications, the difference of path distances may be several or even tens of meters. Hence, a frequency-selective steering channel model is suitable for MMWC \footnote{All the beamforming methods proposed in this paper are only appropriate when the system bandwidth is sufficient to resolve all the multipath components with different steering angles.}, which has been used in \cite{nsenga_2009,xiaozhenyu2013div} and can be expressed as
\begin{equation} \label{eq_channel}
\begin{aligned}
 {\bf{H}}[k] &= \sum\limits_{\ell  = 0}^{L - 1} {{{{\bf{\hat H}}}_\ell }\delta [k - {\tau _\ell }]}= \sqrt {{n_{\rm{r}}}{n_{\rm{t}}}} \sum\limits_{\ell  = 0}^{L - 1} {{{\bf{g}}_\ell }{\lambda _\ell }{\bf{h}}_\ell ^{\rm{H}}\delta [k - {\tau _\ell }]},  \\
 \end{aligned}
\end{equation}
where $
{{{\bf{\hat H}}}_\ell } = \sqrt {{n_{\rm{r}}}{n_{\rm{t}}}} {{\bf{g}}_\ell }{\lambda _\ell }{\bf{h}}_\ell ^{\rm{H}}$, $(\cdot)^{\rm{H}}$ is the conjugate transpose operation, $\delta[k]$ is the discrete impulse response function, $\lambda_\ell$ and $\tau_\ell$ are the channel coefficient and delay of the $\ell$-th MPC, respectively, and ${{\bf{g}}_\ell }$ and ${{\bf{h}}_\ell }$ are receive and transmit steering vectors of the $\ell$-th MPC given by \cite{moraitis_2007,Park_2012_beam_diversity,sayeed_2011,sayeed_2007},
\begin{equation}
{{\bf{g}}_\ell } = \frac{1}{{\sqrt {{n_{\rm{r}}}} }}[{e^{j\pi 0{\Omega _{{\rm{r}}\ell }}}},{e^{j\pi 1{\Omega _{{\rm{r}}\ell }}}},{e^{j\pi 2{\Omega _{{\rm{r}}\ell }}}},...,{e^{j\pi ({n_{\rm{r}}} - 1){\Omega _{{\rm{r}}\ell }}}}]^{\rm{T}},
\end{equation}
and
\begin{equation}
{{\bf{h}}_\ell } = \frac{1}{{\sqrt {{n_{\rm{t}}}} }}[{e^{j\pi 0{\Omega _{{\rm{t}}\ell }}}},{e^{j\pi 1{\Omega _{{\rm{t}}\ell }}}},{e^{j\pi 2{\Omega _{{\rm{t}}\ell }}}},...,{e^{j\pi ({n_{\rm{t}}} - 1){\Omega _{{\rm{t}}\ell }}}}]^{\rm{T}},
\end{equation}
respectively, where $(\cdot)^{\rm{T}}$ is the transpose operator. Note that $\Omega_{{\rm{t}}\ell}$ and $\Omega_{{\rm{r}}\ell}$ represent the cosine transmit and receive angles of the $\ell$-th MPC, and $\Omega_{{\rm{t}}\ell}=\cos(\phi_{{\rm{t}}\ell})$ and $\Omega_{{\rm{r}}\ell}=\cos(\phi_{{\rm{r}}\ell})$, respectively \cite{TseFundaWC}. Therefore, $\Omega_{{\rm{t}}\ell}$ and $\Omega_{{\rm{r}}\ell}$ are within the range $[-1~1)$. For convenience, in the rest of this paper, $\Omega_{{\rm{t}}\ell}$ and $\Omega_{{\rm{r}}\ell}$ are called transmit and receive angles, respectively.

It is assumed that ${{\bf{h}}_\ell }$ and ${{\bf{g}}_\ell }$ are quasi-static at the transmitter and receiver, respectively, which means that ${{\bf{h}}_\ell }$ and ${{\bf{g}}_\ell }$ vary slowly and can be well estimated at the transmitter and receiver, respectively. While $\lambda_\ell$ are independently and identically distributed complex Gaussian variables with zero mean and variance $1/L$. The basic evidence of this assumption is that the multipath directions, which determine ${{\bf{h}}_\ell }$ and ${{\bf{g}}_\ell }$, will change slowly with respect to node or scatter movement, but the phases of the multipath coefficients can change rapidly due to the short wave length of MMWC, leading to a fluctuation/fading in signal strength \cite{geng_2009,Rappaport2002,marinier_1998}, i.e., $|\lambda_\ell|$. It is noted that in slow fading scenario, $\lambda_\ell$ varies slowly compared with the transmission speed. Thus, ${\bf{H}}[k]$ varies slowly and can be well estimated, which means full CSI may be available at both the transmitter and receiver. However, in fast fading scenario, $\lambda_\ell$ varies fast compared with the transmission speed, which means ${\bf{H}}[k]$ varies fast. Thus, the estimation of ${\bf{H}}[k]$ becomes frequent and time consuming, which may greatly degrade the system efficiency.

By exploiting transmit and receive beamformings, the received signals $y[m]$ are expressed as
\begin{equation}
\begin{aligned}
 y[m] &= \sqrt \gamma  \sum\limits_{\ell  = 0}^{L - 1} {{\bf{w}}_{\rm{r}}^{\rm{H}}{{{\bf{\hat H}}}_\ell }{{{{\bf{w}}_{\rm{t}}}}}s[m - {\tau _\ell }]}  + {{\bf{w}}_{\rm{r}}^{\rm{H}}}{\bf{n}}, \\
 \end{aligned}
\end{equation}
where $\gamma$ is the transmit SNR, which refers to the SNR without accounting the array gain achieved by beamforming in this paper, $s[m]$ are the normalized transmitted information symbols, ${{\bf{w}}_{\rm{t}}}$ and ${{\bf{w}}_{\rm{r}}}$ are transmit and receive AWVs with unit 2-norms, respectively, i.e., $\|{{\bf{w}}_{\rm{t}}}\|_2=\|{{\bf{w}}_{\rm{r}}}\|_2=1$, $L$ is the number of MPCs, ${\bf{n}}$ is the standard Gaussian complex noise vector.

The problem is how to design appropriate ${{\bf{w}}_{\rm{t}}}$ and ${{\bf{w}}_{\rm{r}}}$ to optimize the system performance under difference cases. In the next two sections, solutions to this problem under the slow and fast fading channels are proposed, respectively.

\section{Iterative EVD Scheme}
In this section, we first introduce the IEVD scheme which requires full CSI. Afterwards, the training approach to obtain the sub-optimal AWVs in IEVD are described to avoid channel estimation and iterative EVD computation, which makes the proposed IEVD scheme practically feasible. Finally, the analysis on the convergence of IEVD and the training approach is given.

\subsection{Description of IEVD}
For single-layer beamforming, provided full CSI, the optimal transmit and receive AWVs can be found under well-known performance criteria, e.g., maximizing receive SNR. We adopt the receive SNR as the criterion in this context, because we assume that the inter-symbol interference (ISI) caused by MPCs can be well addressed with equalization technology. In the case that there is no equalizer at the receiver after beamforming, signal-to-interference plus noise (SINR) may be a better criterion. With the criterion of receive SNR, under frequency-flat channels, the optimal AWVs are the principal singular vectors of the channel matrix, and they can be obtained by singular-vector decomposition (SVD) on the channel matrix. However, under frequency-selective channels, the optimal AWVs aiming to maximizing receive SNR are difficult to find, and a feasible solution is still not available in the literature to the best of our knowledge. Although Nsenga's scheme \cite{nsenga_2009}, which exploits EVD and Schmidt decomposition within a high-dimensional space tensed by the transmit and receive AWVs, finds a sub-optimal solution, the channel estimation is still time-consuming, and the computations of EVD and Schmidt decomposition in the tensor space are also complicated due to the high dimension. Thus, in this subsection, we propose another sub-optimal solution.

The optimization problem is formulated as
\begin{equation} \label{eq_op_Gamma}
\begin{aligned}
 &{\mathop{\rm maximize}\nolimits}~&\Gamma  = \sum\limits_{\ell  = 0}^{L - 1} {|{\bf{w}}_{\rm{r}}^{\rm{H}}{{{\bf{\hat H}}}_\ell }{{\bf{w}}_{\rm{t}}}{|^2}},  \\
 &{\mathop{\rm {subject~to}}\nolimits} &\|{{\bf{w}}_{\rm{t}}}{\|_2} = \|{{\bf{w}}_{\rm{r}}}{\|_2} = 1, \\
 \end{aligned}
 \end{equation}
where $\Gamma$ is the power gain (or the array gain, and $\gamma\Gamma$ is the receive SNR). It is noted that in MMWC there are two types of antenna arrays in practice that correspond to different constraints on AWV of one-layer beamforming. The first one is phased array with constant amplitude. In such a case all the elements of ${{\bf{w}}_{\rm{t}}}$ or ${{\bf{w}}_{\rm{r}}}$ should be the same constant, and only the phases can be adjusted. For instance, \cite{Ayach2014,Xiaozy2014BeamTrain} have adopted this model. While the other one is antenna array with both amplitude and phase adjustable, which has also be widely used in \cite{xia_2008_prac_ante_traning,Park_2012_beam_diversity,nsenga_2009,xiaozhenyu2013div}, and there are also practical implementations of this model, e.g., in \cite{Emami2011,Pinel2009}. In our model, both amplitude and phase are adjustable, i.e., the latter one is adopted.

In general, when we discuss the optimality of the solution to this problem, we implicitly assume that ${{{\bf{\hat H}}}_\ell }$ is known \emph{a priori}, which corresponds to the slow fading case, where ${{{\bf{\hat H}}}_\ell }$ can be periodically estimated without too much overhead. However, in the fast fading case, the estimation of ${{{\bf{\hat H}}}_\ell }$ becomes frequent and the overhead used for the estimation becomes high, beamforming can only be conduct according to the quasi-static steering vectors, which will be discussed in the next section.

Given that ${{{\bf{\hat H}}}_\ell }$ is known \emph{a priori}, when $L=1$, the optimal AWVs are the principal singular vectors of ${{{\bf{\hat H}}}_0}$. When $L>1$, the optimal AWVs are difficult to find. However, if given ${{\bf{w}}_{\rm{r}}}$, we have
\begin{equation}
\Gamma  = \sum\limits_{\ell  = 0}^{L - 1} {{\bf{w}}_{\rm{t}}^{\rm{H}}{\bf{\hat H}}_\ell ^{\rm{H}}{{\bf{w}}_{\rm{r}}}{\bf{w}}_{\rm{r}}^{\rm{H}}{{{\bf{\hat H}}}_\ell }{{\bf{w}}_{\rm{t}}}}  = {\bf{w}}_{\rm{t}}^{\rm{H}}\left( {\sum\limits_{\ell  = 0}^{L - 1} {{\bf{\hat H}}_\ell ^{\rm{H}}{{\bf{w}}_{\rm{r}}}{\bf{w}}_{\rm{r}}^{\rm{H}}{{{\bf{\hat H}}}_\ell }} } \right){{\bf{w}}_{\rm{t}}}.
 \end{equation}
The optimal ${{\bf{w}}_{\rm{t}}}$ in such a case is the principal eigenvector of $\left( {\sum_{\ell  = 0}^{L - 1} {{\bf{\hat H}}_\ell ^{\rm{H}}{{\bf{w}}_{\rm{r}}}{\bf{w}}_{\rm{r}}^{\rm{H}}{{{\bf{\hat H}}}_\ell }} } \right)$. Similarly, given ${{\bf{w}}_{\rm{t}}}$, we have
\begin{equation}
\Gamma  = \sum\limits_{\ell  = 0}^{L - 1} {{\bf{w}}_{\rm{r}}^{\rm{H}}{{{\bf{\hat H}}}_\ell }{{\bf{w}}_{\rm{t}}}{\bf{w}}_{\rm{t}}^{\rm{H}}{\bf{\hat H}}_\ell ^{\rm{H}}{{\bf{w}}_{\rm{r}}}}  = {\bf{w}}_{\rm{r}}^{\rm{H}}\left( {\sum\limits_{\ell  = 0}^{L - 1} {{{{\bf{\hat H}}}_\ell }{{\bf{w}}_{\rm{t}}}{\bf{w}}_{\rm{t}}^{\rm{H}}{\bf{\hat H}}_\ell ^{\rm{H}}} } \right){{\bf{w}}_{\rm{r}}}.
\end{equation}
The optimal ${{\bf{w}}_{\rm{r}}}$ in such a case is the principal eigenvector of $\left( {\sum_{\ell  = 0}^{L - 1} {{{{\bf{\hat H}}}_\ell }{{\bf{w}}_{\rm{t}}}{\bf{w}}_{\rm{t}}^{\rm{H}}{\bf{\hat H}}_\ell ^{\rm{H}}} } \right)$. Based on this, we suggest the IEVD scheme shown in Algorithm \ref{alg:sub}. This scheme is a typical \emph{alternating optimization} approach \cite{LiQiang2013_AO,Lamare2011_AO}, which does not guarantee the optimal solution, but is efficient to find a sub-optimal solution.

\begin{algorithm}[tb]\caption{The IEVD Scheme.}\label{alg:sub}
\begin{algorithmic}
\STATE \textbf{1) Initialize:}
\begin{quote}
Randomly pick a normalized initial receive AWV ${{\bf{w}}_{\rm{r}}}$.
\end{quote}

\STATE \textbf{2) Iteration:}
\begin{quote}
Iterate the following process $\varepsilon$ times, then stop.
\vspace{0.1 cm}

Compute EVD on $\left( {\sum_{\ell  = 0}^{L - 1} {{\bf{\hat H}}_\ell ^{\rm{H}}{{\bf{w}}_{\rm{r}}}{\bf{w}}_{\rm{r}}^{\rm{H}}{{{\bf{\hat H}}}_\ell }} } \right)$, and set the principal eigenvector to ${{\bf{w}}_{\rm{t}}}$.

\vspace{0.1 cm}
Compute EVD on $\left( {\sum_{\ell  = 0}^{L - 1} {{{{\bf{\hat H}}}_\ell }{{\bf{w}}_{\rm{t}}}{\bf{w}}_{\rm{t}}^{\rm{H}}{\bf{\hat H}}_\ell ^{\rm{H}}} } \right)$, and set the principal eigenvector to ${{\bf{w}}_{\rm{r}}}$.
\end{quote}

\STATE \textbf{3) Result:}
\begin{quote}
${{\bf{w}}_{\rm{t}}}$ is the transmit AWV, and ${{\bf{w}}_{\rm{r}}}$ is the receive AWV.
\end{quote}

\end{algorithmic}
\end{algorithm}

\subsection{The Training Approach}
It is clear that the IEVD scheme relies on \emph{a priori} CSI. Since there are totally $n_{\rm{r}}\times n_{\rm{t}}\times L$ coefficients in the multipath channel matrices, the conventional channel estimation which estimates these coefficients one by one is rather time-consuming. Although when the path number is much smaller than the antenna number, the CSI can be fully characterized by the direction of departure (DoD), direction of arrival (DoA) and fading coefficient of each path, the angle estimation of DoD and DoA may also be complicated. Thus, only in slow fading scenario, the IEVD scheme may be applicable in principle. In fact, even in slow fading scenario, the time-costly channel estimation may still significantly degrade the system efficiency due to the overhead used for channel estimation. In addition, the iterative EVD consumes much computation resource. Thus, the channel estimation and EVD computation make the IEVD scheme not so attractive even in slow fading scenario. To address this problem, we suggest the training approach to obtain the AWVs in the IEVD scheme, which is based on the well-known ``power method'' \cite{Horn_Johnson}, i.e., the principal eigenvector of an arbitrary $N\times N$ matrix $\bf{X}$ can be approximated computed by normalizing ${{\bf{X}}^K}{\bf{w}}$, where $\bf{w}$ is an arbitrary vector, given that $K$ is sufficiently large.


The training approach utilizes the reciprocal feature of the channel, i.e., given ${\bf{H}}[k]$ as the forward channel response from source (the transmitter) to destination (the receiver) \footnote{In the training process, the transmitter also needs to receive with the same antenna array, and the receiver also needs to transmit with the same antenna array. Thus, we use the source and destination instead.}, as shown in (\ref{eq_channel}), the backward channel response from destination to source is $({\bf{H}}[k])^{\rm{H}}$ \footnote{Strictly speaking, the reciprocal channel of $({\bf{H}}[k])$ should be $({\bf{H}}[k])^{\rm{T}}$ \cite{tang_2005}, but $({\bf{H}}[k])^{\rm{H}}$ is also usually used instead for convenience \cite{Bharath2013}. They are equivalent in beamforming design.}. Based on the simplified approach to compute EVD as well as the reciprocal channel, the training approach is introduced in Algorithm \ref{alg:train}. With this training approach, the channel estimation and iterative EVD computation are bypassed, which greatly increases the practical feasibility of the IEVD scheme.

There are various stopping rules for IEVD and its training approach. A simple one is to stop after a certain number of iterations, e.g., $\varepsilon$ iterations used in Algorithms \ref{alg:sub} and \ref{alg:train}. In fact, these two schemes converge fast and basically only need 2 or 3 iterations, which will be shown later. Another one is to compute $\Gamma$ after the $n$-th iteration and get $\Gamma^{(n)}$ according to (\ref{eq_op_Gamma}). When $\Gamma^{(n)}/\Gamma^{(n-1)}<\mu$, stop the iteration, where $\mu$ is a predefined threshold slightly greater than 1, e.g., $\mu=1.05$.

It is noted that the training approach share the same iteration principle with IEVD. Thus, they not only have the same stopping rules, but also have the same convergence properties, which will be analyzed in the next subsection.

\begin{algorithm}[tb]\caption{The Training Approach.}\label{alg:train}
\begin{algorithmic}
\STATE \textbf{1) Initialize:}
\begin{quote}
Randomly pick a normalized initial AWV ${{\bf{w}}_{\rm{r}}}$ at the destination.
\end{quote}

\STATE \textbf{2) Iteration:}
\begin{quote}
Iterate the following process $\varepsilon$ times, then stop.
\vspace{0.1 cm}

\emph{The destination transmits training sequences to the source: }Keep transmitting training sequences with the same AWV ${{\bf{w}}_{\rm{r}}}$ at the destination over $n_{\rm{t}}$ slots, one sequence in a slot. Meanwhile, use identity matrix ${\bf{I}}_{n_{\rm{t}}}$ as the receive AWVs at the source, i.e., the $i$-th column of ${\bf{I}}_{n_{\rm{t}}}$ as the receive AWV at the $i$-th slot. Ignoring the noise and decoding the training sequences, in the $i$-th slot we receive ${r_i}[\ell ] = {\bf{e}}_i^{\rm{H}}{\bf{\hat H}}_\ell ^{\rm{H}}{{\bf{w}}_{\rm{r}}}$, $i = 1,2,...,{n_{\rm{t}}}$, $\ell  = 1,2,...,L$, where ${{\bf{e}}_i}$ is the $i$-th column of ${\bf{I}}_{n_{\rm{t}}}$. Then we arrive at
${\bf{r}}[\ell ] = {\{ {r_i}[\ell ]\} _{i = 1,2,...,{n_{\rm{t}}}}} = {{\bf{I}}_{n_{\rm{t}}}^{\rm{H}}}{\bf{\hat H}}_\ell ^{\rm{H}}{{\bf{w}}_{\rm{r}}} = {\bf{\hat H}}_\ell ^{\rm{H}}{{\bf{w}}_{\rm{r}}}$ and ${{\bf{R}}_{\rm{S}}} = \sum_{\ell  = 0}^{L - 1} {{\bf{r}}[\ell ]{{\left( {{\bf{r}}[\ell ]} \right)}^{\rm{H}}}}$. Normalize ${\bf{R}}_{\rm{S}}^K{{\bf{e}}_1}$ and set the result to ${{\bf{w}}_{\rm{t}}}$ as a new AWV in source.

\vspace{0.1 cm}
\emph{The source transmits training sequences to the destination: }Keep transmitting training sequences with the same AWV ${{\bf{w}}_{\rm{t}}}$ at the source over $n_{\rm{r}}$ slots, one sequence in a slot. Meanwhile, use identity matrix ${\bf{I}}_{n_{\rm{r}}}$ as the receive AWVs at the destination, i.e., the $j$-th column of ${\bf{I}}_{n_{\rm{r}}}$ as the receive AWV at the $j$-th slot. Ignoring the noise and decoding the training sequences, in the $j$-th slot we receive ${{\bar r}_j}[\ell ] = {\bf{\bar e}}_j^{\rm{H}}{{{\bf{\hat H}}}_\ell }{{\bf{w}}_{\rm{t}}}$, $j = 1,2,...,{n_{\rm{r}}}$, $\ell  = 1,2,...,L$, where ${\bf{\bar e}}_j$ is the $j$-th row of ${\bf{I}}_{n_{\rm{r}}}$. Then we arrive at
${\bf{\bar r}}[\ell ] = {\{ {{\bar r}_j}[\ell ]\} _{j = 1,2,...,{n_{\rm{r}}}}} = {{\bf{I}}_{n_{\rm{r}}}^{\rm{H}}}{{{\bf{\hat H}}}_\ell }{{\bf{w}}_{\rm{t}}} = {{{\bf{\hat H}}}_\ell }{{\bf{w}}_{\rm{t}}}$ and ${{\bf{R}}_{\rm{D}}} = \sum_{\ell  = 0}^{L - 1} {{\bf{\bar r}}[\ell ]{{\left( {{\bf{\bar r}}[\ell ]} \right)}^{\rm{H}}}}$. Normalize ${\bf{R}}_{\rm{D}}^K{{{\bf{\bar e}}}_1}$ and set the result to ${{\bf{w}}_{\rm{r}}}$ as a new AWV in destination.
\end{quote}

\STATE \textbf{3) Result:}
\begin{quote}
${{\bf{w}}_{\rm{t}}}$ is the transmit AWV, and ${{\bf{w}}_{\rm{r}}}$ is the receive AWV.
\end{quote}

\end{algorithmic}
\end{algorithm}

\subsection{Convergence Analysis}
The convergence of IEVD, as well as its training approach, will be evident after proving the following theorem.

\begin{theorem}
Let $\Gamma^{(n)}$ denote the value of $\Gamma$ after the $n$-th iteration. Then $\{\Gamma^{(n)}|n=1,2,...\}$ is a non-descending sequence, i.e., $\Gamma^{(n+1)}\geq \Gamma^{(n)}$.
\end{theorem}

\begin{proof}
Firstly, we have
\begin{equation}
\begin{aligned}
{\Gamma ^{(n)}} &= \sum\limits_{\ell  = 0}^{L - 1} {|{{\left( {{\bf{w}}_{\rm{r}}^{(n)}} \right)}^{\rm{H}}}{{{\bf{\hat H}}}_\ell }{\bf{w}}_{\rm{t}}^{(n)}{|^2}} \\
 &= {\left( {{\bf{w}}_{\rm{t}}^{(n)}} \right)^{\rm{H}}}\left( {\sum\limits_{\ell  = 0}^{L - 1} {{\bf{\hat H}}_\ell ^{\rm{H}}{\bf{w}}_{\rm{r}}^{(n)}{{\left( {{\bf{w}}_{\rm{r}}^{(n)}} \right)}^{\rm{H}}}{{{\bf{\hat H}}}_\ell }} } \right){\bf{w}}_{\rm{t}}^{(n)},
 \end{aligned}
 \end{equation}
where ${{\bf{w}}_{\rm{t}}^{(n)}}$ and ${{\bf{w}}_{\rm{r}}^{(n)}}$ are the transmit and receive AWVs after the $n$-th iteration. According to Algorithm 1, ${{\bf{w}}_{\rm{t}}^{(n+1)}}$ is the principal eigenvector of $\left( {\sum_{\ell  = 0}^{L - 1} {{\bf{\hat H}}_\ell ^{\rm{H}}{\bf{w}}_{\rm{r}}^{(n)}{{\left( {{\bf{w}}_{\rm{r}}^{(n)}} \right)}^{\rm{H}}}{{{\bf{\hat H}}}_\ell }} } \right)$, i.e., the optimal solution to maximize ${\Gamma ^{(n)}}$ given ${{\bf{w}}_{\rm{r}}^{(n)}}$. Let
\begin{equation}
\Gamma _0^{^{(n + 1)}} = {\left( {{\bf{w}}_{\rm{t}}^{(n + 1)}} \right)^{\rm{H}}}\left( {\sum\limits_{\ell  = 0}^{L - 1} {{\bf{\hat H}}_\ell ^{\rm{H}}{\bf{w}}_{\rm{r}}^{(n)}{{\left( {{\bf{w}}_{\rm{r}}^{(n)}} \right)}^{\rm{H}}}{{{\bf{\hat H}}}_\ell }} } \right){\bf{w}}_{\rm{t}}^{(n + 1)},
 \end{equation}
then we have $\Gamma _0^{^{(n + 1)}}\geq {\Gamma ^{(n)}}$.

Let us write $\Gamma _0^{^{(n + 1)}}$ in another form
\begin{equation}
\Gamma _0^{^{(n + 1)}} = {\left( {{\bf{w}}_{\rm{r}}^{(n)}} \right)^{\rm{H}}}\left( {\sum\limits_{\ell  = 0}^{L - 1} {{{{\bf{\hat H}}}_\ell }{\bf{w}}_{\rm{t}}^{(n + 1)}{{\left( {{\bf{w}}_{\rm{t}}^{(n + 1)}} \right)}^{\rm{H}}}{\bf{\hat H}}_\ell ^{\rm{H}}} } \right){\bf{w}}_{\rm{r}}^{(n)}.
\end{equation}
Since ${\left( {{\bf{w}}_{\rm{r}}^{(n+1)}} \right)}$ is the principal eigenvector of $\left( {\sum_{\ell  = 0}^{L - 1} {{{{\bf{\hat H}}}_\ell }{\bf{w}}_{\rm{t}}^{(n + 1)}{{\left( {{\bf{w}}_{\rm{t}}^{(n + 1)}} \right)}^{\rm{H}}}{\bf{\hat H}}_\ell ^{\rm{H}}} } \right)$, i.e., the optimal solution to maximize $\Gamma _0^{^{(n + 1)}}$ given ${\left( {{\bf{w}}_{\rm{t}}^{(n+1)}} \right)}$, we have ${\Gamma ^{(n+1)}}\geq \Gamma _0^{^{(n + 1)}}$, because
\begin{equation}
\begin{aligned}
 {\Gamma ^{(n + 1)}} &= \sum\limits_{\ell  = 0}^{L - 1} {|{{\left( {{\bf{w}}_{\rm{r}}^{(n + 1)}} \right)}^{\rm{H}}}{{{\bf{\hat H}}}_\ell }{\bf{w}}_{\rm{t}}^{(n + 1)}{|^2}}  \\
  &= {\left( {{\bf{w}}_{\rm{r}}^{(n + 1)}} \right)^{\rm{H}}}\left( {\sum\limits_{\ell  = 0}^{L - 1} {{{{\bf{\hat H}}}_\ell }{\bf{w}}_{\rm{t}}^{(n + 1)}{{\left( {{\bf{w}}_{\rm{t}}^{(n + 1)}} \right)}^{\rm{H}}}{\bf{\hat H}}_\ell ^{\rm{H}}} } \right){\bf{w}}_{\rm{r}}^{(n + 1)}. \\
 \end{aligned}
 \end{equation}

Therefore, we have ${\Gamma ^{(n+1)}}\geq \Gamma _0^{^{(n + 1)}}\geq {\Gamma ^{(n)}}$.
\end{proof}

According to (\ref{eq_op_Gamma}), $\Gamma<\infty$. Thus, $\{\Gamma^{(n)}|n=1,2,...\}$ converges to a sub-optimal value, which guarantee the convergence of IEVD and its training approach.

The convergence rate of IEVD and its training approach is fast. From the iteration process it can be observed that a temporary sub-optimal transmit and receive AWV can be directly found out in each iteration, rather than progressing by only a small step towards the ultimate directions like Newton's method. Fig. \ref{fig:convergence_examp} shows a simple example, which is a surface of $z=-(x-1)^2-(y-1)^2$. With the \emph{alternating optimization} method, like IEVD, the optimal solution can be quickly achieved. That is, given $x$, to optimize $y$, and given the optimized $y$, to optimize $x$. However, with Newton's method, many iterations are required to reach the optimal solution, because each time only a small step can be made to approach the optimal solution.

\begin{figure}[t]
\begin{center}
  \includegraphics[width=\figwidth cm]{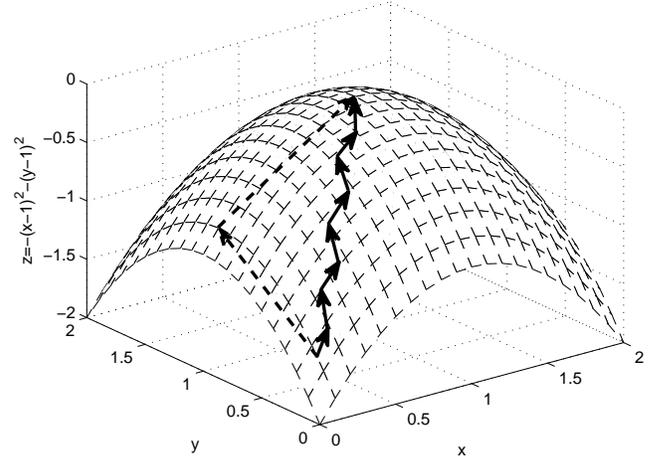}
  \caption{Illustration of the convergence process of IEVD. The arrows with solid and dash lines represent that of Newton's method and IEVD, respectively.}
  \label{fig:convergence_examp}
\end{center}
\end{figure}

\section{Multipath Grouping Scheme}
It is clear that the sub-optimal schemes, including both Nsenga's scheme and the proposed IEVD scheme, require time-consuming estimation of CSI and complicated EVD computation. The training approach for IEVD reduces overhead and computational complexity; thus it is more practical in slow fading scenario where $\lambda_\ell$ varies slowly. However, in fast fading scenario, even the training approach may not be practically applicable, not to mention Nsenga's and IEVD, which are CSI-required, because the training process and CSI estimation will need to be launched frequently due to fast variation of the channel, which will greatly decrease the system efficiency.

In fact, it is not possible to optimize the instantaneous SNR in fast fading scenario due to the unknown instantaneous CSI. In such a case, average SNR may be a reasonable optimization object, but it mingles with diversity gain, i.e., a higher average SNR may lead to a lower diversity gain. We hope to achieve both high array gain and diversity gain, so as to improve both average SNR and robustness of the system.

It is noted that in fast fading scenario, $\lambda_\ell$ varies fast, but ${{\bf{g}}_\ell }$ and ${{\bf{h}}_\ell }$, the receive and transmit steering vectors, can still be well estimated and treated known \emph{a priori}, because they vary slowly as mentioned in Section II. By utilizing this feature, Park and Pan proposed a scheme to achieve diversity gain in addition to array gain, without \emph{a priori} knowledge of $\lambda_\ell$. Although their scheme is simple and achieve full diversity, it may be not efficient in array gain, and may be infeasible when the number of MPCs is larger than that of the antennas at either ends. Hence, we propose an improved diversity scheme for beamforming by exploiting MPG. In this section, we will first briefly introduce Park and Pan's scheme (Park-Pan), and then introduce the proposed MPG scheme.

\subsection{The Diversity Scheme by Park and Pan}
The main beamforming approach for MMWC is to beamform toward only the direction of a single path to achieve array gain, but diversity cannot be achieved. In \cite{Park_2012_beam_diversity}, multiple beams, which steer at different multipath directions, are concurrently shaped at both the transmitter and receiver to achieve beam diversity. It is noted that a beamforming scheme steering at different multipath directions means that the AWVs at both ends are set to have gains on the steering vectors of these MPCs. For instance, if the transmitter and receiver beamform towards only the direction of the $k$-th MPC, the AWVs are ${\bf{w}}_{\rm{t}}={\bf{h}}_k$ and ${\bf{w}}_{\rm{r}}={\bf{g}}_k$, respectively. While if the transmitter and receiver beamform towards multiple MPCs, the gains along the steering vectors of these MPCs need to be set first at the transmitter and receiver, respectively, and then the AWVs can be computed according to these gains.

 Letting ${\alpha _\ell }$ and ${\beta _\ell }$ denote the gains along the directions of the $\ell$-th MPC at the receiver and transmitter, respectively, we have
\begin{equation}
\begin{aligned}
 y[m] = \sqrt {{{{n_{\rm{r}}}{n_{\rm{t}}}\gamma }}} \sum\limits_{\ell  = 0}^{L - 1} {{\alpha _\ell }{\beta _\ell }{\lambda _\ell }s[m - {\tau _\ell }]}  + {{\bf{w}}_{\rm{r}}^{\rm{H}}}{\bf{n}}, \\
 \end{aligned}
\end{equation}
where \begin{equation}
{\alpha _\ell } = {\bf{w}}_{\rm{r}}^{\rm{H}}{{\bf{g}}_\ell },~~{\beta _\ell } = {\bf{h}}_\ell ^{\rm{H}}{{\bf{w}}_{\rm{t}}}.\end{equation}

As in the considered fast fading scenario $\lambda_\ell$ are assumed unavailable at both ends, but ${{\bf{g}}_\ell }$ and ${{\bf{h}}_\ell }$ are known \emph{a priori} at the receiver and the transmitter, respectively, the receive and transmit antenna gains along the directions of different MPCs can be simply and naturally set as ${\alpha _\ell}={\beta _\ell}=1$ for $\ell=0,1,...,L-1$.

Let ${\bf{b}}$ and ${\bf{a}}$ denote the transmit and receive gain vectors, respectively. We have
\begin{equation}
{\bf{b}} = {[{\beta _0},{\beta _1},...,{\beta _{L - 1}}]^{\rm{T}}}
~\text{and}~~
{\bf{a}} = {[{\alpha _0},{\alpha _1},...,{\alpha _{L - 1}}]^{\rm{T}}}.
\end{equation}
The receive and transmit AWVs are consequently achieved as
\begin{equation} \label{eq_awvr}
{{\bf{w}}_{\rm{r}}} = ({{\bf{\bar G}}^{\rm{H}}})_{\rm{R}}^{\dag}\,{{\bf{a}}^ * }={\bf{\bar G}}{\left( {{{\bf{\bar G}}^{\rm{H}}}{\bf{\bar G}}} \right)^{ - 1}}{{\bf{a}}^ * },
\end{equation}
and
\begin{equation}
{{\bf{w}}_{\rm{t}}} = ({{\bf{\bar H}}^{\rm{H}}})_{\rm{R}}^{\dag}\,{\bf{b}}={\bf{\bar H}}{\left( {{{\bf{\bar H}}^{\rm{H}}}{\bf{\bar H}}} \right)^{ - 1}}{\bf{b}},
\end{equation}
respectively, where $(\cdot)^*$ is the conjugation operation, $(\cdot)_{\rm{R}}^{\dag}$ represents the right pseudo-inverse,
\begin{equation}
{\bf{\bar G}} = [{{\bf{g}}_0},{{\bf{g}}_1},...,{{\bf{g}}_{L - 1}}],~\text{and}~
{\bf{\bar H}} = [{{\bf{h}}_0},{{\bf{h}}_1},...,{{\bf{h}}_{L - 1}}].
\end{equation}
The final transmit and receive AWVs require a normalization on the obtained ${{\bf{w}}_{\rm{t}}}$ and ${{\bf{w}}_{\rm{r}}}$, respectively, for the unit 2-norm constraint.

Park-Pan achieves full diversity and simple to implement in the case of $L\leq \min(\{n_{\rm{t}},n_{\rm{r}}\})$. In such a case, ${{\bf{\bar G}}^{\rm{H}}}$ and ${{\bf{\bar H}}^{\rm{H}}}$ are both row-rank matrices, and $({{\bf{\bar G}}^{\rm{H}}})_{\rm{R}}^{\dag}$ and $({{\bf{\bar H}}^{\rm{H}}})_{\rm{R}}^{\dag}$ exist, which guarantees a solution of AWV. However, when $L> \min(\{n_{\rm{t}},n_{\rm{r}}\})$, $({{\bf{\bar G}}^{\rm{H}}})_{\rm{R}}^{\dag}$ or $({{\bf{\bar H}}^{\rm{H}}})_{\rm{R}}^{\dag}$ may not exist, because ${\left( {{{\bf{\bar G}}^{\rm{H}}}{\bf{\bar G}}} \right)^{ - 1}}$ or ${\left( {{{\bf{\bar H}}^{\rm{H}}}{\bf{\bar H}}} \right)^{ - 1}}$ in (\ref{eq_awvr}) may not exist, which means that the solution of AWV may not exist. Therefore, this scheme is not applicable in such a case.

Additionally, Park-Pan may be not effective when the MPCs have close steering angles. To illustrate this, let us look at the equation in the singular-vector space at the transmitter, i.e.,
\[{{\bf{\bar H}}^{\rm{H}}}{{\bf{w}}_{\rm{t}}} = {\bf{b}} \Rightarrow {{{\bf{\bar S}}}^{\rm{H}}}{{{\bf{\bar U}}}^{\rm{H}}}{{\bf{w}}_{\rm{t}}} = {{{\bf{\bar V}}}^{\rm{H}}}{\bf{b}},\]
where ${\bf{\bar H}} = {\bf{\bar U\bar S}}{{{\bf{\bar V}}}^{\rm{H}}}$ is the SVD of $\bf{\bar H}$, ${\bf{\bar U}} = {\{ {{{\bf{\bar u}}}_i}\} _{i = 1,2,...,{n_{\rm{t}}}}}$, ${\bf{\bar V}} = {\{ {{{\bf{\bar v}}}_i}\} _{i = 1,2,...,L}}$, ${\bf{\bar S}} = {\mathop{\rm diag}\nolimits} ([{\sigma _1},{\sigma _2},...,{\sigma _L}])$. Letting ${{{\bf{\tilde w}}}_{\rm{t}}} = {{{\bf{\bar U}}}^{\rm{H}}}{{\bf{w}}_{\rm{t}}} = {\{ {{\tilde w}_{{\rm{t}}i}}\} _{i = 1,2,...,{n_{\rm{t}}}}}$ and ${\bf{\tilde b}} = {{{\bf{\bar V}}}^{\rm{H}}}{\bf{b}} = {\{ {{\tilde b}_i}\} _{i = 1,2,...,L}}$, we have ${{\tilde w}_{{\rm{t}}i}} = {{\tilde b}_i}/\sigma _i^*,\;\;i = 1,2,...,L$. Note that $|{{\tilde w}_{{\rm{t}}i}}|^2$ and $\sigma _i^2$ denotes the transmission power allocated in the dimension of ${{\bf{\bar u}}}_i$ and the power gain in this dimension, respectively. Naturally, in order to achieve diversity gain and array gain, the transmission power should be evenly allocated in multiple dimensions with good gains. However, for Park-Pan, the power allocation is ineffective, because less power is allocated in the dimension with a greater gain $\sigma_i$. When the steering angles of MPCs are largely spaced, the steering vectors are approximately orthogonal. Thus, the condition number of ${\bf{\bar H}}$, i.e., $\sigma_1/\sigma_L$, is small, and Park-Pan achieves good performance. However, when MPCs have close steering vectors, the condition number of ${\bf{\bar H}}$ will be very large, which means most of the power is allocated in the dimension of ${{\bf{\bar u}}}_L$, which has the smallest gain. In such a case, the achieved array gain will be poor.

\subsection{MPG Scheme}
It is found that Park-Pan may be infeasible when there are too many MPCs, and may be not effective when there are MPCs with close steering angles, which makes the condition number of ${\bf{\bar H}}$ bad. Consequently, it probably benefits to group the MPCs with close steering angles to shape a single equivalent MPC, which will maintain a good condition number, because the equivalent MPCs are basically widely spaced. This is the motivation of the MPG scheme. As this equivalent MPC represents all the MPCs within the corresponding group in AWV computations, the MPG operation can not only avoid the ineffectiveness caused by these MPCs with close angles, but also reduce the number of involved MPCs. In this subsection, we present the MPG beamforming scheme. As shown in Fig. \ref{fig:scheme}, in our scheme the MPCs are grouped according to their transmit and receive steering angles at the transmitter and receiver, respectively. For each non-empty group, an equivalent steering vector is defined, and the corresponding antenna gain is set.

\begin{figure}[t]
\begin{center}
  \includegraphics[width=7.6 cm]{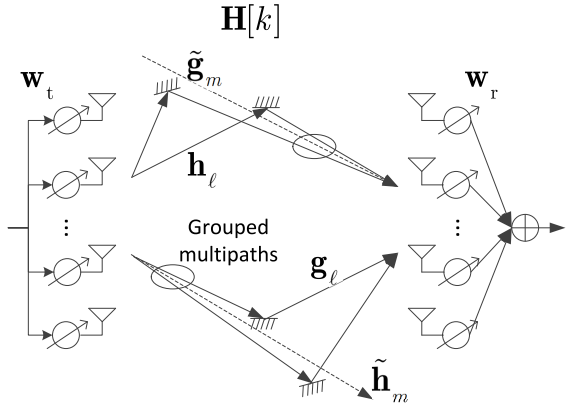}
  \caption{Illustration of the MPG schemes.}
  \label{fig:scheme}
\end{center}
\end{figure}

\subsubsection{Transmit MPG}
An MPG scheme includes both transmit and receive MPCs. Let us look at the transmitter first. The angle range of the transmit steering vector is $[-1~1)$. Uniformly divide the angle range into $n_{\rm{t}}$ segments with a spacing $2/n_{\rm{t}}$, where the angle range of the $i$-th segment is $[-1+{2}(i-1)/{n_{\rm{t}}}~-1+{2}i/{n_{\rm{t}}})$. As each MPC corresponds to a transmit steering vector, they are grouped into $n_{\rm{t}}$ groups according to their steering angles ${{\Omega _{{\rm{t}}\ell }}}$. If ${{\Omega _{{\rm{t}}\ell }}}$ fall within the $i$-th angle segment, the corresponding MPCs are grouped into the $i$-th group. For each group that has at least one MPC in it, i.e., non-empty, an equivalent steering vector is defined. Subsequently, when calculating the transmit AWV, each non-empty group is represented by its corresponding equivalent steering vector.

As an equivalent steering vector represents all the MPCs within the corresponding group in AWV computations, it must have a significant correlation with these MPCs, i.e., the inner product between the equivalent steering vector and the steering vector of any one MPC in this group is significant. There are many approaches to define the equivalent steering vectors for the non-empty groups. For example, a natural way is to define a new steering vector with a steering angle being the average angle of the steering angles of the MPCs within this group. However, this method requires estimation of all the steering angles of the MPCs, which is with high complexity and not practical. Therefore, we suggest an approach that is natural but simple to implement to define the equivalent steering vectors.

Let $\{N^{(i)}_j|j=1,2,...,N_i\}$ denote the indices of the transmit steering vectors ${\bf{h}}_\ell$ falling within the $i$-th non-empty group, the equivalent steering vector for this group is defined as

\begin{equation}
{{{\bf{\tilde h}}}_i} = {\left({\sum\limits_{j = 1}^{{N_i}} {{{\bf{h}}_{N_j^{(i)}}}} }\right)}\Big{/}{{\Big{\|}\sum\limits_{j = 1}^{{N_i}} {{{\bf{h}}_{N_j^{(i)}}}\Big{\|}}_2 }}.
\end{equation}
The corresponding antenna gain in the direction of this equivalent steering vector is simply set as ${{\tilde \beta }_i} = 1$.

For the groups that have no multipath in them, i.e., empty, there are no equivalent steering vectors or antenna gains. Hence, the total number of equivalent transmit steering vectors, $N_{\rm{T}}$, is no larger than $n_{\rm{t}}$. Lastly, the transmit AWV is obtained as
\begin{equation} \label{eq_wt}
{{\bf{w}}_{\rm{t}}} = {\bf{\tilde H}}{\left( {{{{\bf{\tilde H}}}^{\rm{H}}}{\bf{\tilde H}}} \right)^{ - 1}}{\bf{\tilde b}},
\end{equation}
where
\begin{equation}
{\bf{\tilde H}} = [{{{\bf{\tilde h}}}_0},{{{\bf{\tilde h}}}_1},...,{{{\bf{\tilde h}}}_{{N_{\rm{T}}} - 1}}],
\end{equation}
and
\begin{equation}{\bf{\tilde b}} = {[{{\tilde \beta }_0},{{\tilde \beta }_1},...,{{\tilde \beta }_{N_{\rm{T}} - 1}}]^{\rm{T}}}=[1,1,...,1]^{\rm{T}}.
\end{equation}
The final transmit AWV requires a normalization on the obtained ${{\bf{w}}_{\rm{t}}}$ for the unit 2-norm constraint.

It is noted that as $N_{\rm{T}}\leq n_{\rm{t}}$, ${( {{{{\bf{\tilde H}}}^{\rm{H}}}{\bf{\tilde H}}} )^{ - 1}}$ exists almost surely. In other words, the transmit multipath grouping leads to a solution of transmit AWV almost surely.

\subsubsection{Receive MPG}
Analogously, at the receiver we also conduct multipath grouping.  Uniformly divide the angle range into $n_{\rm{r}}$ segments with a spacing $2/n_{\rm{r}}$, where the angle range of the $i$-th segment is $[-1+{2}(i-1)/{n_{\rm{r}}}~-1+{2}i/{n_{\rm{r}}})$. As each MPC also corresponds to a receive steering vector, they are grouped into $n_{\rm{r}}$ groups according to their steering angles ${{\Omega _{{\rm{r}}\ell }}}$. If ${{\Omega _{{\rm{r}}\ell }}}$ fall within the $i$-th angle segment, the corresponding MPCs are grouped into the $i$-th group. For each group that has at least one MPC in it, i.e., non-empty, an equivalent steering vector is similarly defined. Let $\{M^{(i)}_j|j=1,2,...,M_i\}$ denote the indices of the receive steering vectors ${\bf{g}}_\ell$ falling within the $i$-th non-empty group, the equivalent steering vector for this group is defined as
\begin{equation}
{{{\bf{\tilde g}}}_i} = {\left({\sum\limits_{j = 1}^{{M_i}} {{{\bf{g}}_{M_j^{(i)}}}} }\right)}\Big{/}{{\Big{\|}\sum\limits_{j = 1}^{{M_i}} {{{\bf{g}}_{M_j^{(i)}}}\Big{\|}_2} }}.
\end{equation}
The corresponding antenna gain in the direction of this equivalent steering vector is simply set as ${{\tilde \alpha }_i} = 1$.

For the groups that have no multipath in them, there are no equivalent steering vectors or antenna gains. Hence, the total number of equivalent transmit steering vectors, $M_{\rm{R}}$, is no larger than $n_{\rm{r}}$. Lastly, the receive AWV is obtained as
\begin{equation}
{{\bf{w}}_{\rm{r}}} = {\bf{\tilde G}}{\left( {{{{\bf{\tilde G}}}^{\rm{H}}}{\bf{\tilde G}}} \right)^{ - 1}}{\bf{\tilde a}}^*,
\end{equation}
where
\begin{equation}
{\bf{\tilde G}} = [{{{\bf{\tilde g}}}_0},{{{\bf{\tilde g}}}_1},...,{{{\bf{\tilde g}}}_{{M_{\rm{R}}} - 1}}],
\end{equation}
and
\begin{equation}
{\bf{\tilde a}} = {[{{\tilde \alpha }_0},{{\tilde \alpha }_1},...,{{\tilde \alpha }_{M_{\rm{R}} - 1}}]^{\rm{T}}}=[1,1,...,1]^{\rm{T}}.
\end{equation}
The final receive AWV requires a normalization on the obtained ${{\bf{w}}_{\rm{r}}}$ for the unit 2-norm constraint.

 Likewise, since $M_{\rm{R}}\leq n_{\rm{r}}$, ${( {{{{\bf{\tilde G}}}^{\rm{H}}}{\bf{\tilde G}}} )^{ - 1}}$ exists almost surely, which leads to a solution of receive AWV almost surely.

\subsubsection{Realization} The realization of the MPG scheme seems to require angle estimation at the transmitter and the receiver, which is complicated. In fact, the peak of the spatial spectrum from a Bartlett beamformer can be used instead of more complicated DoD and DoA estimations. Take the transmit MPG for instance. Predefine $n_{\rm{t}}$ unit vectors ${\bf{v}}_i$, $i=1,2,...,n_{\rm{t}}$:
\begin{equation}
\begin{aligned}
{{\bf{v}}_i}=& \frac{1}{{\sqrt {{n_{\rm{t}}}} }}[{e^{j\pi 0*(-1+(2i-1)/n_{\rm{t}})}},{e^{j\pi 1*(-1+(2i-1)/n_{\rm{t}})}},\\
&{e^{j\pi 2*(-1+(2i-1)/n_{\rm{t}})}},...,{e^{j\pi ({n_{\rm{t}}} - 1)*(-1+(2i-1)/n_{\rm{t}})}}]^{\rm{T}},
\end{aligned}
\end{equation}
For an arbitrary MPC with transmit angle ${{\Omega _{{\rm{t}}\ell }}}$, its angle locates in the $i$-th segment, i.e., $[-1+{2}(i-1)/{n_{\rm{t}}}~-1+{2}i/{n_{\rm{t}}})$, is equivalent to that its steering vector ${{\bf{h}}_\ell }$ satisfies $|{\bf{v}}_i^{\rm{H}}{{\bf{h}}_\ell }| > |{\bf{v}}_j^{\rm{H}}{{\bf{h}}_\ell }{|_{j \ne i}}$. Hence, the alternative approach to realize MPG is to estimate the steering vectors ${{\bf{h}}_\ell }$, $\ell=1,2,...,L$, at the transmitter, and for each steering vector ${{\bf{h}}_\ell }$, find the index of the unit vectors that ${\mathop{\rm maximizes}\nolimits} \;\mathop {\arg }\limits_i \;|{\bf{v}}_i^{\rm{H}}{{\bf{h}}_\ell }|$. The index is the group number of ${{\bf{h}}_\ell }$. The receive MPC can be realized in the same way.

There are many ways to estimate the transmit and receive steering vectors, which are much simpler than to estimate full CSI. One way is to estimate them by the first iteration process of Algorithm \ref{alg:train}. In the first iteration process, we can obtain ${\bf{r}}[\ell ]$ and ${\bf{\bar r}}[\ell ]$, respectively. In fact,
\begin{equation}
{\bf{r}}[\ell ] = {\bf{\hat H}}_\ell ^{\rm{H}}{{\bf{w}}_{\rm{r}}} = \left( {\sqrt {{n_{\rm{r}}}{n_{\rm{t}}}} {\lambda _\ell }{\bf{g}}_\ell ^{\rm{H}}{{\bf{w}}_{\rm{r}}}} \right){{\bf{h}}_\ell },
\end{equation}
and
\begin{equation}
{\bf{\bar r}}[\ell ] = {{{\bf{\hat H}}}_\ell }{{\bf{w}}_{\rm{t}}} = \left( {\sqrt {{n_{\rm{r}}}{n_{\rm{t}}}} {\lambda _\ell }{\bf{h}}_\ell ^{\rm{H}}{{\bf{w}}_{\rm{t}}}} \right){{\bf{g}}_\ell }.
\end{equation}
Thus, by normalizing ${\bf{r}}[\ell ]$ and ${\bf{\bar r}}[\ell ]$, we get ${{\bf{h}}_\ell }$ and ${{\bf{g}}_\ell }$ at the transmitter (the source in Algorithm \ref{alg:train}) and receiver (the destination in Algorithm \ref{alg:train}), respectively.

It is noted that to use MPG in practice there are many practical system issues to be considered, and a critical one is array calibration. To do accurate spectrum estimation, the phase and amplitude responses of each antenna need to be known. Although these could be measured ahead of time since there is only one active RF chain and the multiple channels are passive, the calibration on the responses of these channels can be a rather tricky issue because the gain is adjustable on each channel. As array calibration is not the focus of this paper, relevant literatures, such as \cite{Keizer2011} and \cite{Ng2009}, are referred for further considerations.

\subsection{Performance Analysis}
Since PEP is an extensively used metric to reflect both array gain and diversity gain \cite{TseFundaWC}, we adopt it in this context. It is clear that no matter whether MPG or Park-Pan is adopted, an equivalent SISO multipath fading channel is observed after the single-layer beamforming, which yields
\begin{equation} \label{eq_siso}
h[k] = \sum\limits_{m = 0}^v {{h_m}\delta [k - m]},
\end{equation}
where $v \geq\max (\{ {\tau _\ell }|\ell=0,1,...,L-1\} )$, and
\begin{equation}
{h_m} = \left\{ \begin{aligned}
 &{\bf{w}}_{\rm{r}}^{\rm{H}}{{{\bf{\hat H}}}_\ell }{{\bf{w}}_{\rm{t}}},~&m = {\tau _\ell }, \\
 &0,&{\rm{otherwise}}. \\
 \end{aligned} \right.
 \end{equation}

With the above SISO model (\ref{eq_siso}), single-carrier block transmission, such as the zero-padded (ZP) block transmission \cite{ohno2006performance}, is adopted to evaluate the system performance. With the ZP, the minimum mean square error (MMSE) receiver can be used to collect the multipath diversity \cite{ohno2006performance}, and the pairwise-error probability (PEP) is upper bounded as
\begin{equation} \label{eq_pep}
\begin{aligned}
&P\{ {{\bf{s}}_A} \to {{\bf{s}}_B}\} \lessapprox {\mathbb{E}_{\lambda_\ell}}\left( {Q\left( {\frac{{d\sqrt {\gamma \sum_{\ell  = 0}^{L - 1} {|{h_{{\tau _\ell }}}{|^2}} } }}{{\sqrt {2{N_0}} }}} \right)} \right)\\
=& {\mathbb{E}_{\lambda_\ell}}\left( {Q\left( {d\sqrt {\frac{\gamma }{2}\sum\limits_{\ell  = 0}^{L - 1} {|{\bf{w}}_r^H{{\bf{g}}_\ell }{\lambda _\ell }{\bf{h}}_\ell ^{\rm{H}}{{\bf{w}}_t}{|^2}} } } \right)} \right),
\end{aligned}
\end{equation}
where ${\mathbb{E}_{\rm{\lambda_\ell}}}$ is the expectation on the channel coefficients $\lambda_\ell$ which are fast fading, $Q(x)$ is the Q function, $d$ is the minimum distance of the signal constellation. Note that the ZP based block transmission achieves the best performance among all block based transmission systems for an SISO channel \cite{ohno2006performance}.

Since when $a>1$, the Q function $Q(a)$ has the upper bound $e^{-a^2/2}$ \cite{TseFundaWC}, when $\gamma$ is large enough, we further have
\begin{equation} \label{eq_pep1}
\begin{aligned}
 P\{ {{\bf{s}}_A} \to {{\bf{s}}_B}\}  &\le {\mathbb{E}_{{\lambda _\ell }}}\left( {\exp \left( { - {d^2}\frac{\gamma }{4}\sum\limits_{\ell  = 0}^{L - 1} {|{\bf{w}}_r^H{{\bf{g}}_\ell }{\lambda _\ell }{\bf{h}}_\ell ^{\rm{H}}{{\bf{w}}_t}{|^2}} } \right)} \right) \\
  &= {\mathbb{E}_{{\lambda _\ell }}}\left( {\prod\limits_{\ell  = 0}^{L - 1} {\exp \left( { - \frac{{{d^2}\gamma }}{4}|{\bf{w}}_r^H{{\bf{g}}_\ell }{\lambda _\ell }{\bf{h}}_\ell ^{\rm{H}}{{\bf{w}}_t}{|^2}} \right)} } \right) \\
  &= \prod\limits_{\ell  = 0}^{L - 1} {{\mathbb{E}_{{\lambda _\ell }}}\left( {\exp \left( { - \frac{{{d^2}\gamma }}{4}|{\bf{w}}_r^H{{\bf{g}}_\ell }{\bf{h}}_\ell ^{\rm{H}}{{\bf{w}}_t}{|^2}|{\lambda _\ell }{|^2}} \right)} \right)}  \\
  &= \prod\limits_{\ell  = 0}^{L - 1} {\left( {\frac{1}{{1 + {d^2}\gamma |{\bf{w}}_r^H{{\bf{g}}_\ell }{\bf{h}}_\ell ^{\rm{H}}{{\bf{w}}_t}{|^2}/4}}} \right)}\buildrel \Delta \over = {P_{{\rm{UB}}}}, \\
\end{aligned}
\end{equation}
which is suitable for both MPG and Park-Pan. Note that the difference in the AWV setting will lead to different PEP performance between these two schemes.

According to the definition in \cite{TseFundaWC}, the diversity gain $D$ can thus be derived as
\begin{equation} \label{eq_diversity}
\begin{aligned}
 D &=  - \mathop {\lim }\limits_{\gamma  \to \infty } \frac{{\log {P_{{\rm{UB}}}}}}{{\log \gamma }} \\
  &=  - \mathop {\lim }\limits_{\gamma  \to \infty } \frac{{\log \prod\limits_{\ell  = 0}^{L - 1} {\left( {\frac{1}{{1 + {d^2}\gamma |{\bf{w}}_r^H{{\bf{g}}_\ell }{\bf{h}}_\ell ^{\rm{H}}{{\bf{w}}_t}{|^2}/2}}} \right)} }}{{\log \gamma }} \\
  &= L, \\
\end{aligned}
\end{equation}
which means that both MPG and Park-Pan achieve full diversity.

\section{Performance Comparison}
In this section, we conduct extensive comparisons between IEVD and Nsenga's scheme, which both are sub-optimal schemes that require real-time full or partial CSI, as well as MPG and Pank-Pan, which both are simpler schemes feasible even in fast fading scenario with only the steering vectors known \emph{a priori}. This section is organized as follows. First, we investigate the convergence rate of the proposed IEVD scheme, as well as that of its training approach. Afterwards, we compare the performances of these schemes and discuss the complexity issue.

\subsection{Convergence Rates of IEVD and its Training Approach}

The convergence rates of IEVD and its training approach are shown in Fig. \ref{fig:convergence} with $n_{\rm{t}}=n_{\rm{r}}=8$. Recall that $L$ is the number of MPCs, and $K$ is the power of the square matrices to approximate the corresponding EVD. These curves are obtained via simulations over $10^5$ randomly realized channels according to the model in (\ref{eq_channel}). In these realizations, the steering angles $\Omega_{{\rm{t}}\ell}$ and $\Omega_{{\rm{r}}\ell}$, which determine the steering vectors, obey a uniform distribution within $[-1~1)$; while the MPC coefficients $\{\lambda_\ell\}$ follow a complex Gaussian distribution with zero mean and variance $1/L$. In each channel realization, we obtain an instantaneous power gain $\Gamma$. The final average power gain is computed by averaging the $10^5$ instantaneous power gains.

It is observed from Fig. \ref{fig:convergence} that both IEVD and its training approach converge fast. As the initial AWVs are randomly generated, there is basically no power gain when the number of iterations is 0. When $L=1$, i.e., there is only one significant MPC, both IEVD and its training approach converge with only 1 iteration, even with $K=1$. When $L>1$, IEVD converges with only two iterations; while its training approach converges depending on $K$. It is clear that the training approach converges faster when $K$ is greater, which means a higher computational complexity. Fortunately, when $K=2$, the training approach can basically achieve convergence with 2 iterations. A small $K$ and a fast convergence rate increase the practically applicability of the training approach.

On the other hand, it is found that as $L$ increases, the obtained average gain becomes smaller, which means that the array gain becomes lower. However, as $L$ increases, more diversity gain is achieved. The metric of power gain cannot reflect diversity gain. Thus, we need to adopt alternative metrics to evaluate the performance.

\begin{figure}[{\rm{t}}]
\begin{center}
  \includegraphics[width=\figwidth cm]{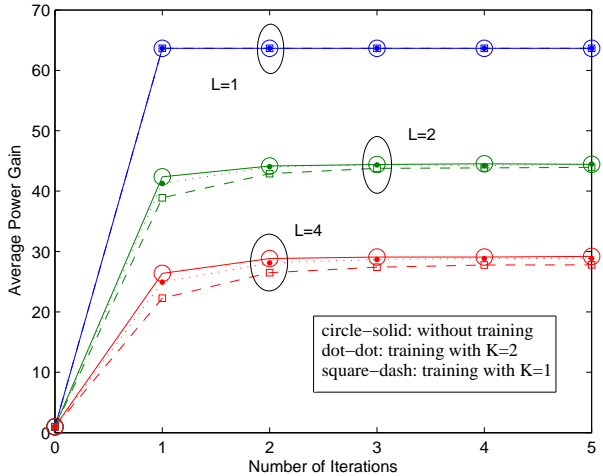}
  \caption{Convergence rates of IEVD and its training approach. $n_{\rm{t}}=n_{\rm{r}}=8$. $L$ is the number of MPC, and $K$ is the power of the square matrices to approximate the corresponding EVD.}
  \label{fig:convergence}
\end{center}
\end{figure}

\subsection{Performance Comparisons}


We next want to obtain and compare the numerical results of the upper bounded PEP for the four involved schemes, namely IEVD, Nsenga's, MPG, and Park-Pan, respectively. It is noted that the training approach of IEVD has the same performance with IEVD in the convergence state given that $K$ is large enough, e.g., $K\geq 2$, and the iteration number is 3 in all the simulations in this subsection. The PEP upper bounds are achieved by Monte Carlo simulation under different scenarios according to (\ref{eq_pep}), which is more accurate than (\ref{eq_pep1}). The steering angles $\Omega_{{\rm{t}}\ell}$ and $\Omega_{{\rm{r}}\ell}$ are either deterministic or random. For deterministic steering angles, only the channel coefficients $\lambda_\ell$ are random. Each realization of $\lambda_\ell$ determines an instantaneous bounded PEP according to (\ref{eq_pep}), and an average bounded PEP is obtained by adopting their mean. For random steering angles, each realization of $\Omega_{{\rm{t}}\ell}$, $\Omega_{{\rm{r}}\ell}$ and $\lambda_\ell$ determines an instantaneous bounded PEP, and an average bounded PEP (termed PEP for short hereafter) is achieved by the same way. In the simulations, $n_{\rm{t}}=n_{\rm{r}}=8$. Quadrature phase shift keying (QPSK) modulation is adopted; thus, $d=\sqrt{2}$.

We first investigate the case that the MPCs have the same transmit steering angle but different receive steering angles. This case corresponds to the scenario that the transmitter is far away from the receiver, and there are many reflectors around the receiver. Afterwards, we investigate the case that the MPCs have different transmit and receive steering angles. This case corresponds to the scenario that the transmitter is not far away from the receiver with many reflectors nearby.

Figs. \ref{fig:fig1} and \ref{fig:fig2} depict the PEP comparisons between the schemes with the same single steering angle at the transmitter and variable steering angles at the receiver, with deterministic and random steering angles, respectively. In the deterministic case, the steering angles are equally spaced, i.e., ${\Omega _{{\rm{t}}\ell }} = {\Omega _{{\rm{r}}\ell }}=  - 1 + {2}(\ell  - 1)/{L}$. While in the random case, the steering angles obey a uniform distribution within $[-1~1)$.

Regarding the case of deterministic steering angles, all the four schemes achieve full diversity. The PEP performance of MPG is exactly the same as that of Park-Pan, as shown in Fig. \ref{fig:fig1}. This is because when $L\leq\min(\{n_{\rm{t}},n_{\rm{r}}\})$, all the MPCs fall within different groups at the receiver, and thus there is actually no such operations that multiple MPCs are grouped into a single one; consequently, the equivalent receive steering vectors of MPG are the same as that of Park-Pan. Also, there is in fact no MPG operation at the transmitter, because the MPCs has the same one transmit steering angle. On the other hand, the PEP performance of IEVD is exactly the same as that of Nsenga's, and achieves a better array gain than that of MPG and Park-Pan, which shows that both sub-optimal schemes are effective. When $L=1$, the four schemes have the same performance, which is optimal, because there is only one steering angle at the transmitter and receiver.

Regarding the case of random steering angles, all the four schemes also achieve full diversity. It is observed that, in such a case, the MPG scheme shows superiority over Park-Pan in array gain, as shown in Fig. \ref{fig:fig2}. As the steering angles are randomly generated at the receiver, MPCs with close receive steering angles are grouped into a single equivalent MPC in the MPG scheme. Thus, the superiority of MPG versus Park-Pan indicates that the multipath grouping operation achieves an array gain, and as $L$ increases, the superiority becomes more significant. On the other hand, IEVD also achieves the same performance as Nsenga's, and has an increasingly better array gain than MPG and Park-Pan.

\begin{figure}[{\rm{t}}]
\begin{center}
  \includegraphics[width=\figwidth cm]{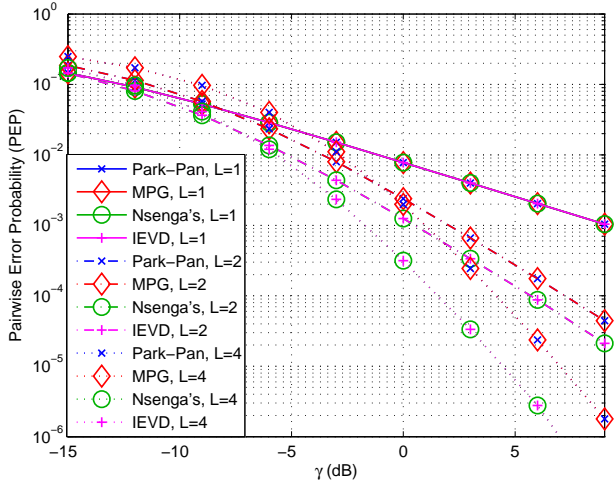}
  \caption{PEP comparisons between the schemes with one same single steering angle at the transmitter and variable steering angles at the receiver. The steering angles at the receiver are deterministic and equally spaced.}
  \label{fig:fig1}
\end{center}
\end{figure}

\begin{figure}[{\rm{t}}]
\begin{center}
  \includegraphics[width=\figwidth cm]{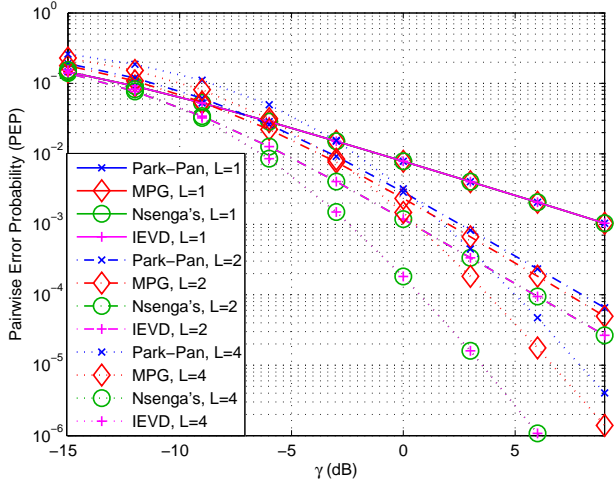}
  \caption{PEP comparisons between the schemes with one same single steering angle at the transmitter and variable steering angles at the receiver. The steering angles at the receiver are random and uniformly distributed within $[-1~1)$.}
  \label{fig:fig2}
\end{center}
\end{figure}

Figs. \ref{fig:fig3} and \ref{fig:fig4} depict the PEP comparisons between the schemes with variable steering angles at the transmitter and receiver, with deterministic and random steering angles, respectively. It is found that in both deterministic and random cases, the same results can be observed as that from Figs. \ref{fig:fig1} and \ref{fig:fig2}, respectively. That is, all the schemes achieve full diversity. Besides, MPG and Park-Pan achieve the same performance in the deterministic case, because there is actually no MPG operation in the MPG scheme. But in the random case, MPG achieves a better array gain than Park-Pan, which benefits from the MPG operation. IEVD and Nsenga's achieve the same performance in both cases, and have an increasingly superiority over MPG and Park-Pan as $L$ increases. Comparing Fig. \ref{fig:fig1} with Fig. \ref{fig:fig3}, as well as Fig. \ref{fig:fig2} with Fig. \ref{fig:fig4}, it is found that the two optimal schemes almost achieve the same performance in the same one transmit steering angle case and the variable transmit steering angles case, but MPG and Park-Pan have an increasing loss in array gain as $L$ increases, which further shows the effectiveness of the sub-optimal schemes.

\begin{figure}[{\rm{t}}]
\begin{center}
  \includegraphics[width=\figwidth cm]{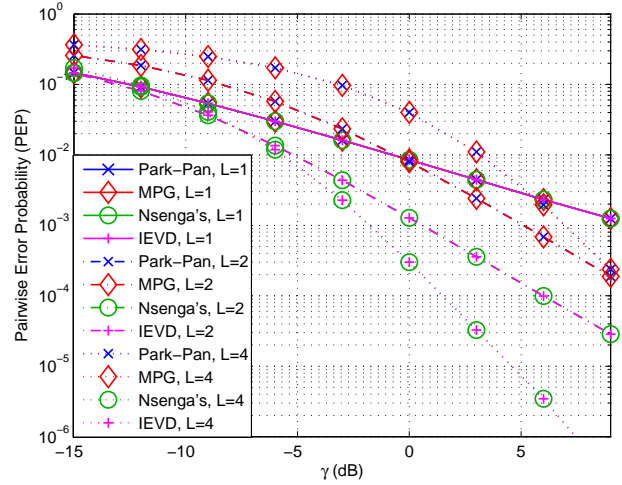}
  \caption{PEP comparisons between the schemes with variable steering angles at the transmitter and receiver. The steering angles are deterministic and equally spaced at both ends.}
  \label{fig:fig3}
\end{center}
\end{figure}

\begin{figure}[{\rm{t}}]
\begin{center}
  \includegraphics[width=\figwidth cm]{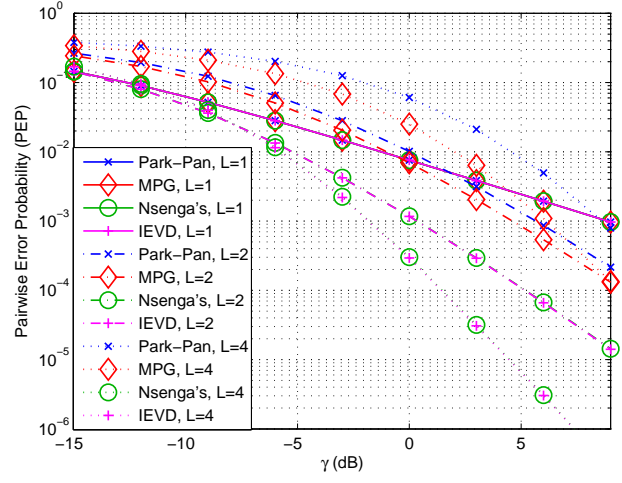}
  \caption{PEP comparisons between the schemes with variable steering angles at the transmitter and receiver. The steering angles are random and uniformly distributed within $[-1~1)$ at both ends.}
  \label{fig:fig4}
\end{center}
\end{figure}

These results are all obtained through the bounded PEP curves. To demonstrate the rational of them, the block-error rate (BLER) curves are obtained via simulation with deterministic and random steering angles as shown in Figs. \ref{fig:fig5} and \ref{fig:fig6}, respectively. In the simulations, the block size is 32, and the ZP length is 8. A single BLER is obtained based on the simulation of $10^7$ transmissions and receptions of a block with randomly realized channels, which guarantees that the BLER curves are precise. From these two figures the same results can be observed and concluded as that from Figs. \ref{fig:fig3} and \ref{fig:fig4}.

\begin{figure}[{\rm{t}}]
\begin{center}
  \includegraphics[width=\figwidth cm]{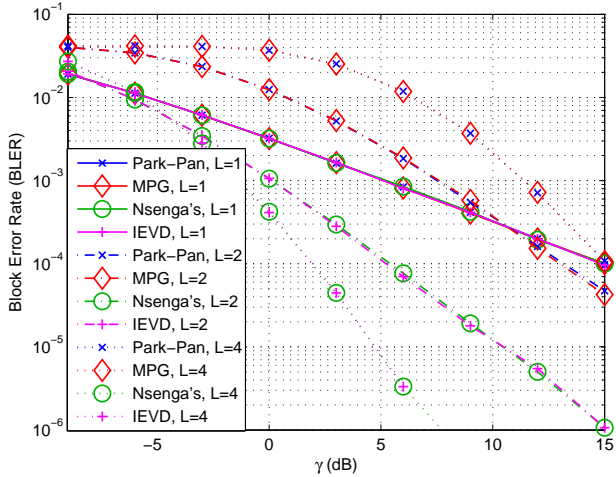}
  \caption{BLER comparisons between the schemes with variable steering angles at the transmitter and receiver. The steering angles are deterministic and equally spaced at both ends.}
  \label{fig:fig5}
\end{center}
\end{figure}

\begin{figure}[{\rm{t}}]
\begin{center}
  \includegraphics[width=\figwidth cm]{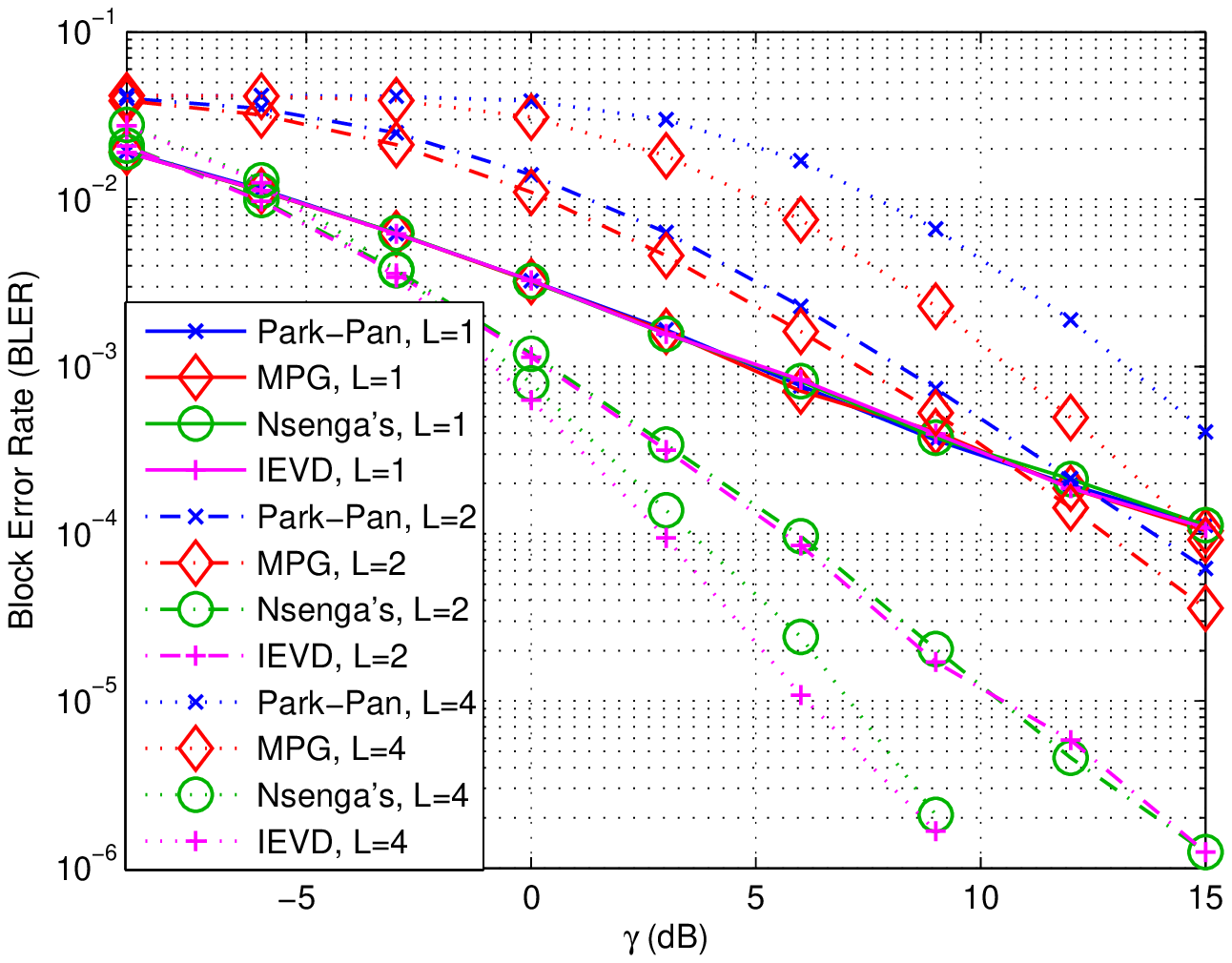}
  \caption{BLER comparisons between the schemes with variable steering angles at the transmitter and receiver. The steering angles are random and uniformly distributed within $[-1~1)$ at both ends.}
  \label{fig:fig6}
\end{center}
\end{figure}

\begin{figure}[{\rm{t}}]
\begin{center}
  \includegraphics[width=\figwidth cm]{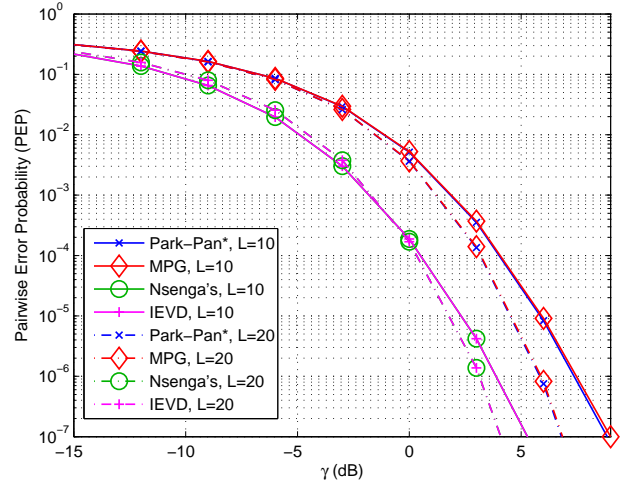}
  \caption{PEP comparisons between the schemes with variable steering angles at the transmitter and receiver, where the number of MPCs is large. The steering angles are deterministic and equally spaced at both ends. Park-Pan* denotes Park-Pan with the process that randomly selecting $n_{\rm{t}}$ and $n_{\rm{r}}$ MPCs at the transmitter and receiver to compute the AWVs.}
  \label{fig:fig11L}
\end{center}
\end{figure}

\begin{figure}[{\rm{t}}]
\begin{center}
  \includegraphics[width=\figwidth cm]{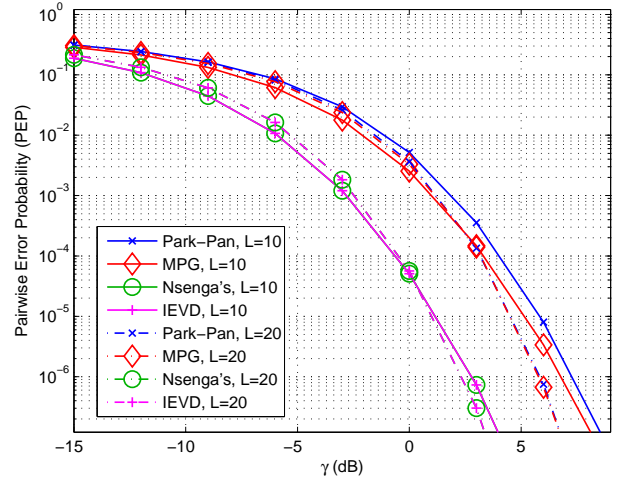}
  \caption{PEP comparisons between the schemes with variable steering angles at the transmitter and receiver, where the number of MPCs is large. The steering angles are random and uniformly distributed within $[-1~1)$ at both ends. Park-Pan* denotes Park-Pan with the process that randomly selecting $n_{\rm{t}}$ and $n_{\rm{r}}$ MPCs at the transmitter and receiver to compute the AWVs.}
  \label{fig:fig12L}
\end{center}
\end{figure}

As aforementioned, Park-Pan is infeasible when $L>\min(\{n_{\rm{t}},n_{\rm{r}}\})$, but the two sub-optimal schemes and MPG are applicable. A natural way to make Park-Pan feasible even when $L>\min(\{n_{\rm{t}},n_{\rm{r}}\})$ is to randomly select $n_{\rm{t}}$ and $n_{\rm{r}}$ MPCs at the transmitter and receiver to compute the AWVs. Park-Pan with such a process is termed as Park-Pan*. Figs. \ref{fig:fig11L} and \ref{fig:fig12L} show their PEP performances in the case of $L>\min(\{n_{\rm{t}},n_{\rm{r}}\})$, with deterministic equally-spaced and uniformly-distributed random steering angles, respectively. It is observed again that in both cases the two sub-optimal schemes achieve the same performance, which is better than that of MPG and Park-Pan*. With deterministic angles, as shown in Fig. \ref{fig:fig11L}, MPG almost achieves the same performance as Park-Pan*. While with random angles, as shown in Fig. \ref{fig:fig12L}, MPG achieves a better array gain in the case of $L=10$, and an approximately equivalent array gain to Park-Pan* in the case of $L=20$. In addition, both MPG and Park-Pan* achieve full diversity. These results show that, in presence of massive MPCs, MPG is still feasible and effective. In addition, Park-Pan can be made feasible by randomly selecting $n_{\rm{t}}$ and $n_{\rm{r}}$ MPCs, instead of $L$ MPCs, at the transmitter and receiver to compute the AWVs.

\subsection{Complexity Issue}
As aforementioned, the two sub-optimal schemes, i.e., IEVD and Nsenga's, require real-time CSI. Thus, they are only applicable in slow fading scenario. In fast fading scenario, the MPG and Park-Pan schemes can be adopted, which both require only the quasi-static steering vectors ${\bf{h}}_\ell$ and ${\bf{g}}_\ell$, rather than the fast varying coefficients $\lambda_\ell$, at the transmitter and receiver, respectively.

Compared with Nsenga's, the proposed IEVD achieves exactly the same performance, but has a reduced overhead and implementation complexity. The overhead of channel estimation for Nsenga's scheme is $n_{\rm{r}}\times n_{\rm{t}}$ training sequences \cite{nsenga_2009}, while that for IEVD with the training approach is $2(n_{\rm{r}}+ n_{\rm{t}})$ training sequences, where the iteration number is set to 2. It is clear that the training overhead is greatly reduced by IEVD with the training approach, especially when the number of antennas is large at the transmitter and receiver. Moreover, IEVD with the training approach only needs to compute matrix multiplication, which has a much lower complexity than Nsenga's scheme, where the computations of EVD and Schmidt decomposition in the tensor space are required. In brief, IEVD with the training approach achieves the same performance with Nsenga's scheme, with a reduced overhead and complexity.

On the other hand, from the performance evaluations it can be found that, compared with Park-Pan, MPG achieves a better array gain. Besides, MPG does not increase overhead, because both MPG and Park-Pan only require the quasi-static steering vectors at the transmitter and receiver. In addition, although MPG requires grouping operation and computations of equivalent steering vectors, which are not required by Park-Pan, it needs a lower-dimensional matrix inverse than Park-Pan, i.e., ${( {{{{\bf{\tilde H}}}^{\rm{H}}}{\bf{\tilde H}}} )^{ - 1}}$ and ${( {{{{\bf{\tilde G}}}^{\rm{H}}}{\bf{\tilde G}}} )^{ - 1}}$ have lower dimensions than ${\left( {{{\bf{\bar H}}^{\rm{H}}}{\bf{\bar H}}} \right)^{ - 1}}$ and ${\left( {{{\bf{\bar G}}^{\rm{H}}}{\bf{\bar G}}} \right)^{ - 1}}$, respectively, due to the MPG operation. In summary, MPG achieve better performance than Park-Pan with an approximately equivalent complexity.

\section{Conclusion}
Two Tx/Rx joint beamforming schemes have been proposed for MMWC in this paper. The first one is IEVD, which is sub-optimal and suitable under slow fading channel. To make this scheme practically feasible, a training approach has been suggested. It is demonstrated that IEVD and its training approach basically converge with only 2 iterations. Compared with the existing sub-optimal scheme, i.e., Nsenga's, IEVD with its training approach achieves the same performance but a reduced overhead and computational complexity. The other one is MPG suitable under fast fading channel, which groups the MPCs and concurrently beamforms towards multiple steering angles of the grouped MPCs. Compared with the existing alternative scheme, i.e., Park-Pan, MPG can work in the case that the number of MPCs is greater than that of antennas at the transmitter or receiver, but Park-Pan cannot. Additionally, MPG achieves full diversity, the same as Park-Pan, but a better array gain than Park-Pan, when the number of MPCs is smaller than or around that of antennas at the transmitter and the receiver, with an approximately equivalent complexity.

\section*{Acknowledgment}

The authors would like to thank the editor and the anonymous reviewers for
their many useful and detailed comments that have helped to improve the
presentation of this manuscript.


\vspace{-0.1 in}
\begin{biography}[{\includegraphics[width=1in,height=1.25in,clip,keepaspectratio]{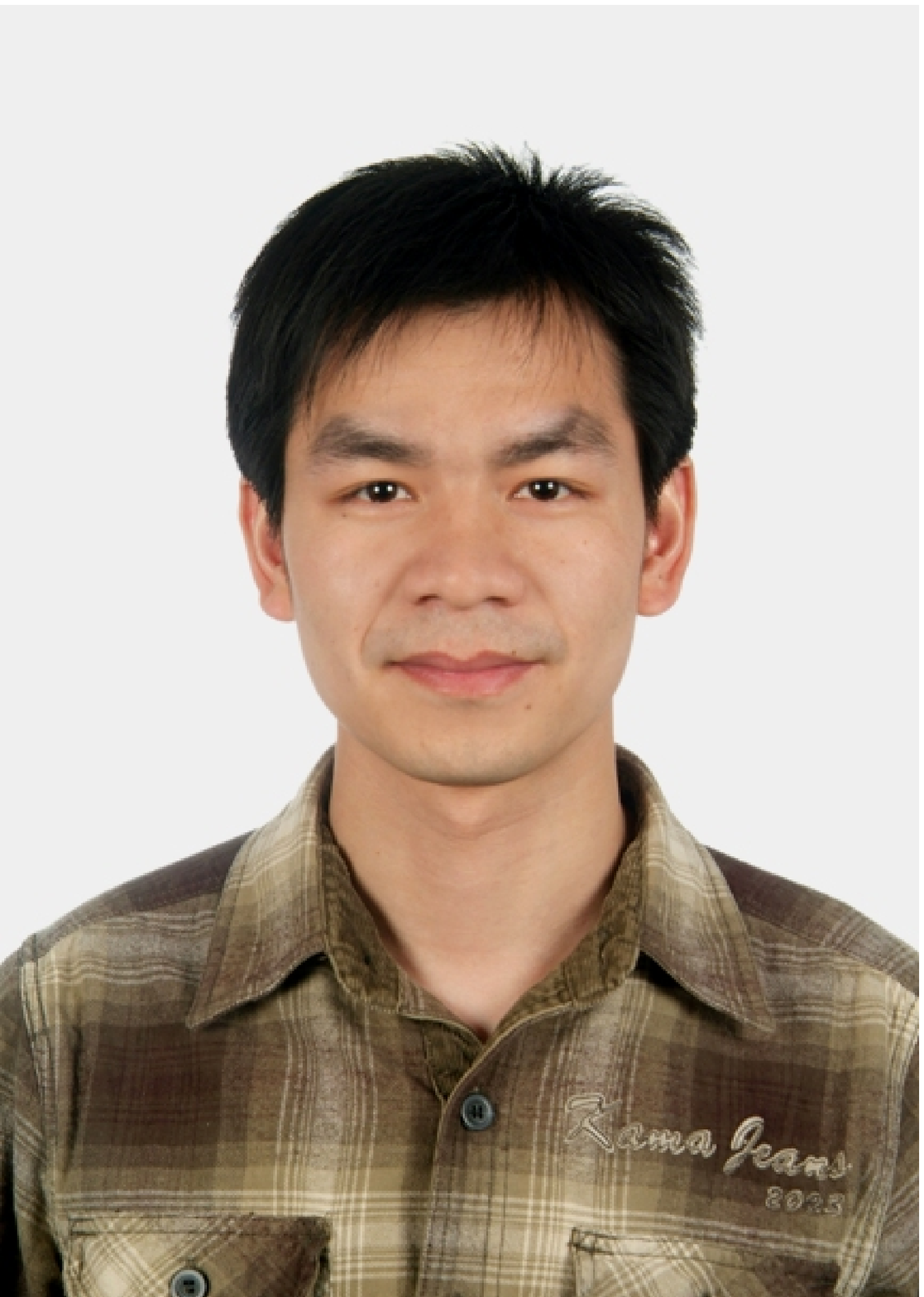}}]{Zhenyu Xiao}
received the Ph.D. degree in the department of Electronic
Engineering from Tsinghua University, Beijing, China, in 2011, and
B.E. degree in the department of Electronics and Information
Engineering from Huazhong University of Science and Technology,
Wuhan, China, in 2006. From 2011 to 2013, he served as a post doctor in the E.E.
department of Tsinghua University, Beijing, China. Since 2013, he has been a lecturer in Beihang University, Beijing, China.

Dr. Xiao has published over 40 papers, and served as reviewers for IEEE Transactions on Signal Processing, IEEE Transactions on Wireless Communications, IEEE Transactions on Vehicular Technology, IEEE Communications Letters, etc. He has been TPC members of IEEE GLOBECOM'12, IEEE WCSP'12, IEEE ICC'15, etc. His current research
interests are communication signal processing and practical system
implementation for wideband communication systems, which cover
synchronization, multipath signal processing, diversity, multiple
antenna technology, etc. Currently his is dedicated in millimeter communications and full-duplex communications.
\end{biography}

\begin{biography}[{\includegraphics[width=1in,height=1.25in,clip,keepaspectratio]{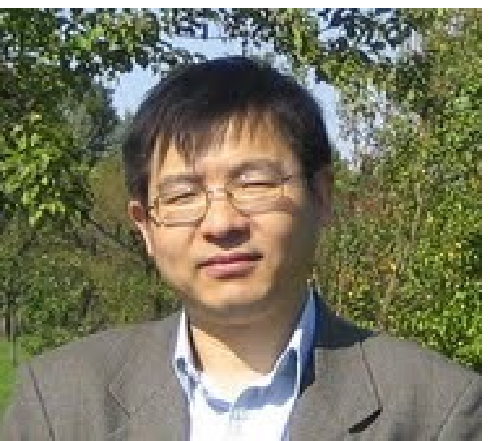}}]{Xiang-Gen Xia} (M'97, S'00, F'09) received his B.S. degree in mathematics from Nanjing Normal University, Nanjing, China, and his M.S. degree in mathematics from Nankai University, Tianjin, China, and his Ph.D. degree in electrical engineering  from the University of Southern California, Los Angeles, in 1983, 1986, and 1992, respectively.

He was a Senior/Research Staff Member at Hughes Research Laboratories, Malibu, California, during 1995-1996. In September 1996, he joined the Department of Electrical and Computer Engineering, University of Delaware, Newark, Delaware, where he is the Charles Black Evans Professor. His current research interests include space-time coding, MIMO and OFDM systems, digital signal processing, and SAR and ISAR imaging. Dr. Xia has over 280 refereed journal articles published and accepted, and 7 U.S. patents awarded and is the author of the book Modulated Coding for Intersymbol Interference Channels (New York, Marcel Dekker, 2000).

Dr. Xia received the National Science Foundation (NSF) Faculty Early Career Development (CAREER) Program Award in 1997, the Office of Naval Research (ONR) Young Investigator Award in 1998, and the Outstanding Overseas Young Investigator Award from the National Nature Science Foundation of China in 2001. He also received the Outstanding Junior Faculty Award of the Engineering School of the University of Delaware in 2001. He is currently serving and has served as an Associate Editor for numerous international journals including IEEE Transactions on Signal Processing, IEEE Transactions on Wireless Communications, IEEE Transactions on Mobile Computing, and IEEE Transactions on Vehicular Technology. Dr. Xia is Technical Program Chair of the Signal Processing Symp., Globecom 2007 in Washington D.C. and the General Co-Chair of ICASSP 2005 in Philadelphia.
\end{biography}

\begin{biography}[{\includegraphics[width=1in,height=1.25in,clip,keepaspectratio]{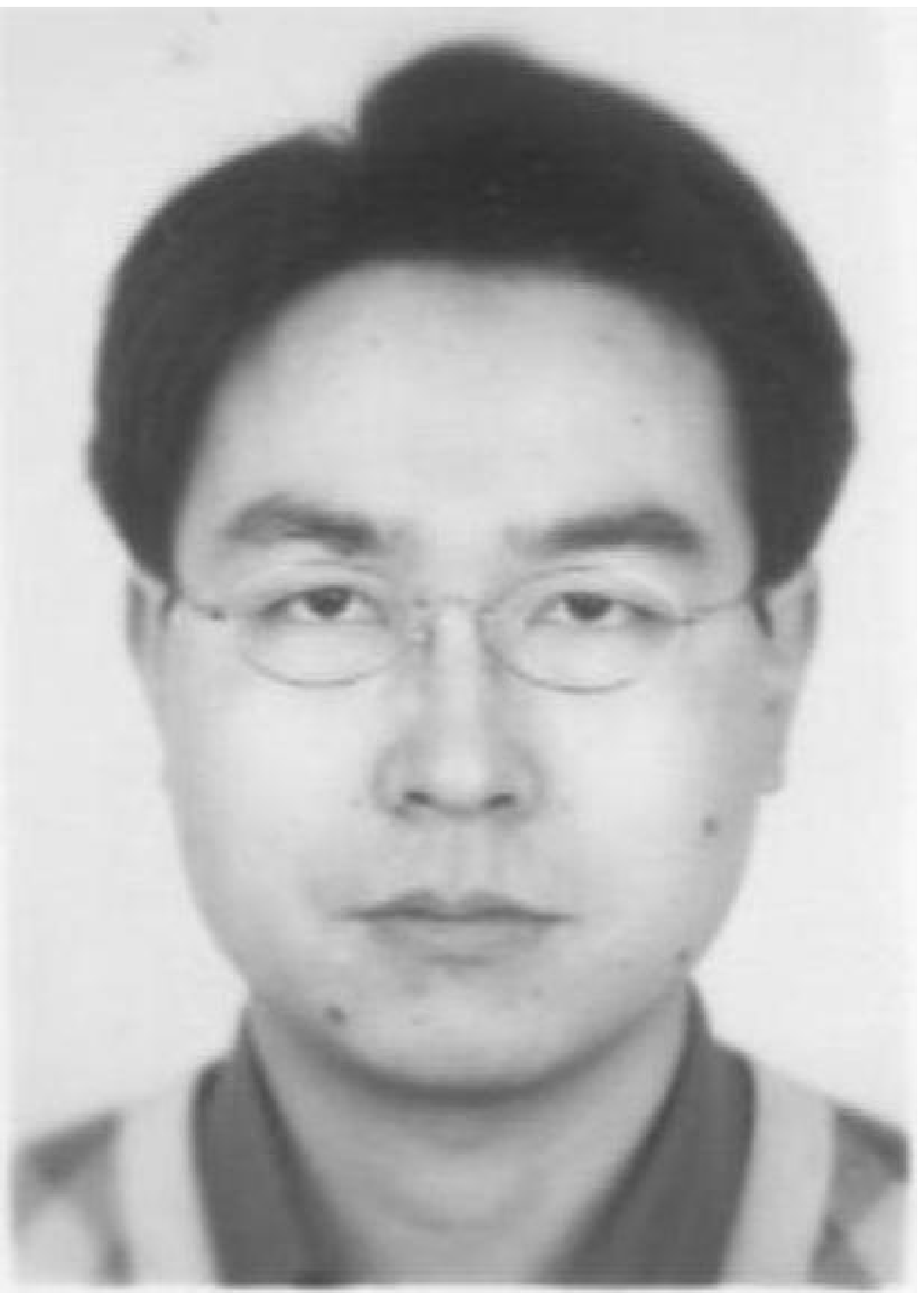}}]{Depeng Jin}
received the B.S. and Ph.D. degrees from Tsinghua University,
Beijing, China, in 1995 and 1999, respectively, both in electronics
engineering. Now he is a professor at Tsinghua University and chair
of Department of Electronic Engineering. Dr. Jin was awarded
National Scientific and Technological Innovation Prize (Second
Class) in 2002. His research fields include telecommunications,
high-speed networks, future internet architecture and ASIC design
for wireless communications.
\end{biography}

\begin{biography}[{\includegraphics[width=1in,height=1.25in,clip,keepaspectratio]{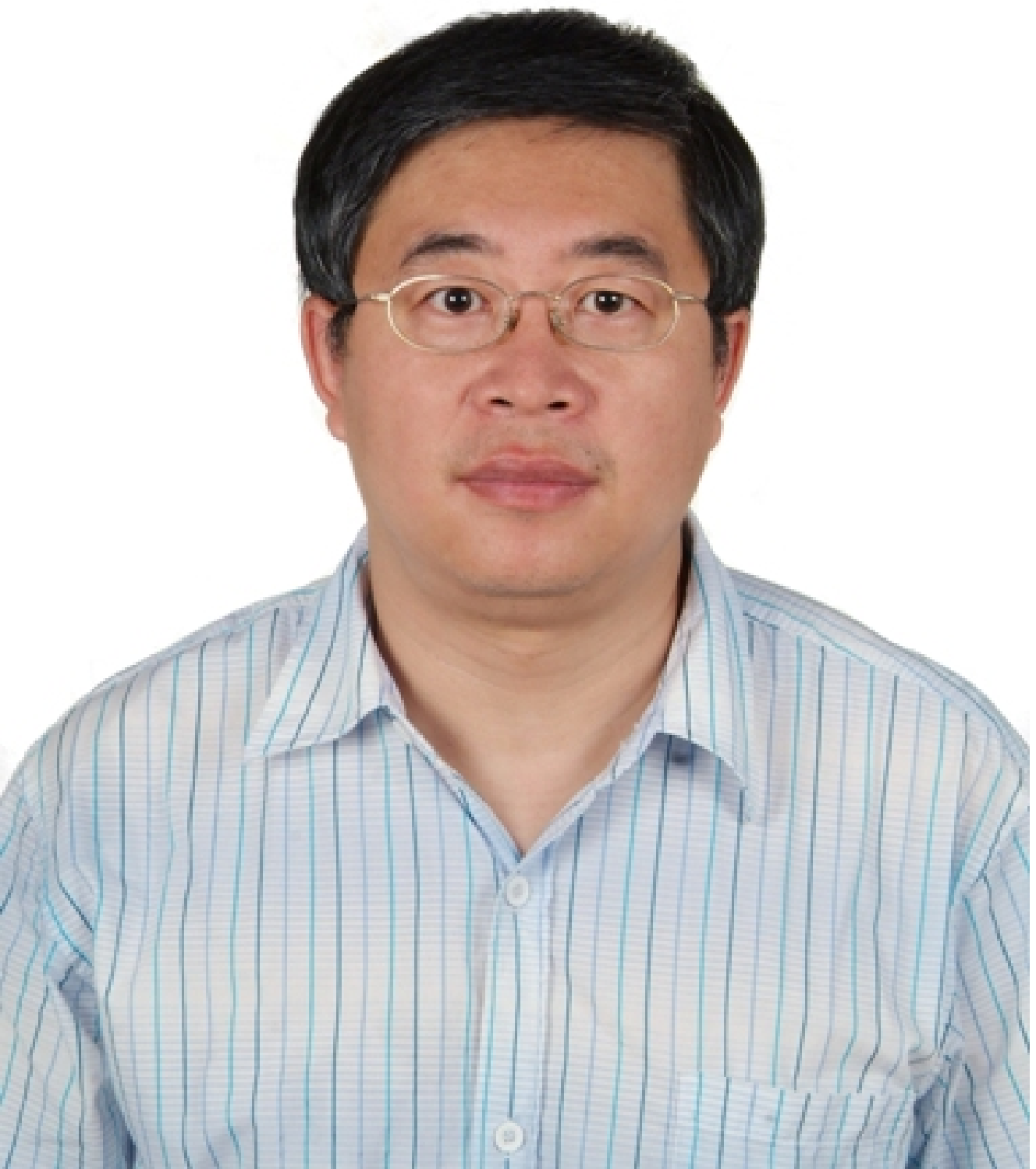}}]{Ning Ge}
received his B.E. degree in 1993, and his Ph.D. in 1997, both from
Tsinghua University, China. From 1998 to 2000, he worked on the
development of ATM switch fabric ASIC in ADC Telecommunications,
Dallas. Since 2000 he has been in the Dept. of Electronics
Engineering at Tsinghua University. Currently he is a professor and
Director of Communication and Microwave Institute. His current
interests are in the areas of communication ASIC design, wireless
communications, etc.
\end{biography}

\end{document}